\documentclass[conference]{IEEEtran}
\IEEEoverridecommandlockouts
\usepackage{cite}
\usepackage{amsmath,amssymb,amsfonts}
\usepackage{algorithm}
\usepackage{algorithmicx}
\usepackage[noend]{algpseudocode}
\usepackage{graphicx}
\usepackage{subfigure}
\usepackage{amsthm}
\newtheorem{theorem}{Theorem}

\usepackage{hyperref}
\usepackage{textcomp}
\usepackage{comment}
\usepackage[frozencache,cachedir=.]{minted}
\usepackage{xcolor}
\def\BibTeX{{\rm B\kern-.05em{\sc i\kern-.025em b}\kern-.08em
    T\kern-.1667em\lower.7ex\hbox{E}\kern-.125emX}}
\begin{document}

\title{Minerva: Decentralized Collaborative Query Processing over InterPlanetary File System
\thanks{Code of Minerva is available:
\href{https://github.com/Fudan-MediaNET/Minerva}{https://github.com/Fudan-MediaNET/Minerva}}
}

\author{Zhiyi Yao, Bowen Ding, Qianlan Bai, Yuedong Xu
\IEEEcompsocitemizethanks{\IEEEcompsocthanksitem Zhiyi Yao, Bowen Ding and Yuedong Xu are with School of Information Science and Technology, Fudan University. E-mail: \{ydxu\}@fudan.edu.cn
}
}

\maketitle

\begin{abstract} Data silos create barriers in accessing and utilizing data dispersed over networks. Directly sharing data easily suffers from the long downloading time, the single point failure and the untraceable data usage. In this paper, we present Minerva, a peer-to-peer cross-cluster data query system based on InterPlanetary File System (IPFS). Minerva makes use of the distributed Hash table (DHT) lookup to pinpoint the locations that store content chunks. We theoretically model the DHT query delay and introduce the fat Merkle tree structure as well as the DHT caching to reduce it. We design the query plan for read and write operations on top of Apache Drill that enables the collaborative query with decentralized workers. 
We conduct comprehensive experiments on Minerva, and the results show that Minerva achieves up to $2.08 \times$ query performance acceleration compared to the original IPFS data query, and could complete data analysis queries on the Internet-like environments within an average latency of $0.615$ second. With collaborative query, Minerva could perform up to $1.39 \times$ performance acceleration than centralized query with raw data shipment.



\end{abstract}


\section{Introduction}
Efficient query and analysis of large data can provide accurate data-driven solutions to a wide range of services and thereby improve productivity. With the development of distributed database and cloud object storage, large-scale data is often stored in private cloud clusters or database. Although this data storage mode ensures the privacy, location and access efficiency of data, it also causes the data isolation problem. Especially for transportation\cite{zheng2010geolife}, medical and academic\cite{article}, the data of various organizations and companies is often stored in a decentralized manner, which makes it difficult to take full advantage of joint big data analysis. 

To alleviate the problem of data silo\cite{10.1145/3429252}, existing methods provide solutions from the perspective of federated queries\cite{FA2022}. Apache Presto\cite{inproceedings} is a distributed query engine that supports different database backend. Presto query data from heterogeneous data sources and accesses data from different data federations\cite{inproceedingspre}. PXF\cite{Raghavan2019PlatformEF} is a query framework for heterogeneous data sources, which realizes the joint query of heterogeneous data sources, and uniformly processes data by defining remote data tables in the form of PXF. Where there is a high requirement for data privacy, secure multi-party computation (SMC)\cite{inproceedingsSMC} can provide secure federated data query and data federation integration. Although all the above methods provide convenience for data sharing and joint query, there are also some problems that prevents practice. These methods either require high speed network to centralize data processing, or require data intermediate structure to process data in a unified form, which increases the cost of data transformation and computation. 

In this paper, we propose Minerva, a new decentralized solution for data federation query and joint open data sharing platform. The core idea of Minerva is to use the file system over the network as the storage of federated data, query and compute the back-end for cross network and cross clusters federated data query. This simple but effective idea has brought many advantages: 1) The unified network file system is used for query, simplifying the cumbersome configuration caused by complex data storage distribution (like multiple data sources or clusters), and providing the possibility of cross cluster data access. 2) Data access is associated with the network, and multiple data federations can be established by building a virtual private network. 3) Data storage nodes are used for distributed query, making full use of computing resources and reducing network transmission load. Through these advantages, multi-party can easily establish a federated data query platform and use shared data for advanced.

To achieve cross cluster data analysis sharing and fine-grained efficient federated queries in the internet, we choose the Interplanetary Filesystem (IPFS)\cite{DBLP:journals/corr/Benet14} as the storage back-end of Minerva system. IPFS is a blockchain based, permanent and decentralized network file system, as well as a point-to-point distributed protocol. IPFS can locate the nodes of data distribution within the logarithmic time complexity. IPFS file version control and file anonymity also guarantee data security. For high-performance SQL parsing and distributed result aggregation, we use Apache drill\cite{articleDRILL} as the SQL parser of Minerva.

The contributions and main work of this paper are summarized as follows:

\begin{itemize}
    \item We have built a cross-network and cross-cluster distributed data query system, which can both serve the shared open data joint query and distributed query of private data. To the best of our knowledge, Minerva is the first system to use native SQL language to conduct decentralized data query and using the network as the storage backend.
    \item Minerva is a high-performance, scalable system to latency. It uses native SQL language to jointly query network data in a cloud independent environment. Minerva can support queries in various data formats, such as csv, json, etc., and track data usage in real time.
    \item We design a decentralized data query mechanism based on IPFS, and transmit compute instead of the original data between nodes, which improve the efficiency of distributed query. We build theoretical model of transmission and computation latency distribution.
    \item We improve the decentralized file index structure suitable for big data queries and proposed our file parsing method suitable for decentralized and DHT-based storage. We demonstrate the effectiveness of our decentralized file parsing by theory and experiments. 
    \item We conduct experiments to evaluate the performance of Minerva. Evaluation shows in distributed network queries, Minerva can achieve high query performance in the WAN, and also has high query planning capability in the local cluster (with latency less than $456$ ms and $615$ ms with $100 Mpbs$ network). the results show that Minerva achieves up to $2.08 \times$ query performance acceleration compared to the original IPFS data query. With collaborative query, Minerva could perform up to $1.39 \times$ performance acceleration than centralized query with raw data shipment.
\end{itemize}


In the rest of this paper, we introduce the work Minerva based on in Sec.\ref{sec:background}. We present an overview of Minerva in Sec.\ref{sec:overview} and elaborate on its design details in Sec.\ref{sec:design}. Then we will deduce and demonstrate our IPFS file resolution scheme in Sec.\ref{sec:delay anatic}. Finally, we present the implementation in Sec.\ref{sec:implementation} and the evaluations in Sec.\ref{sec:evaluation}, and concludes in Sec.\ref{sec:conclusion}.

\section{Background and Motivation}
\label{sec:background}

In this section, we introduce the background knowledge of Apache Drill and IPFS that provide the foundation of our decentralized query system.


\subsection{Apache drill}
Apache Drill\cite{articleDRILL} is a distributed SQL execution engine using Massively parallel processing (MPP) architecture\cite{10.1145/1953122.1953148,8187252}. Drill does not limit the storage format and location of the queried data, and provides a set of abstract operational interfaces so that the SQL parser and the query optimizer are independent of particular storage systems such as Hadoop File System (HDFS)\cite{5496972}, Amazon S3 and MangoDB\cite{articleMongo}. Different from other distributed SQL execution engine like Impala\cite{Kornacker2015ImpalaAM} and Presto\cite{inproceedings}, an important feature of Drill is that it does not need to define the schema of the data source in advance. 

Drill can automatically identify the schema during query execution and build the schema from the data, making it convenient to query semi-structured data formats such as JSON. Although Apache drill runs distributed queries with multiple machines in MMP mode, each node is peer-to-peer for drill and can opt in or out of queries processing. This also provides foundation implementing federated queries of Minerva. 

\subsection{IPFS}
\label{subsec:IPFS}
IPFS\cite{DBLP:journals/corr/Benet14, 22IPFSAndFriends} is a P2P network-based, decentralized file system that is used by blockchain systems as a data storage solution for its decentralized data store and its benefits in terms of security, privacy, and reliability\cite{9684521, 22sigcommIPFS}. IPFS stores data in interconnected data objects, each uniquely identified by its cryptographic hash, called the object's content identifier (CID). Data objects could form a range of data structures such as files, directories, hyperlinked graphs, etc. All the nodes running on the IPFS form a P2P network. When any node acquire data, based on the identifier of the desired data object, it initiates a request to the node where this piece of data is stored, and thus acquires the data. This enables decentralized data access and decentralized applications based on technologies such as blockchain to implement the task of data access through IPFS. IPFS also provides distributed data storage in areas such as the Internet of Things, cloud computing, and literature and data sharing.

\begin{figure}[htbp]
\centerline{\includegraphics[scale=0.38]{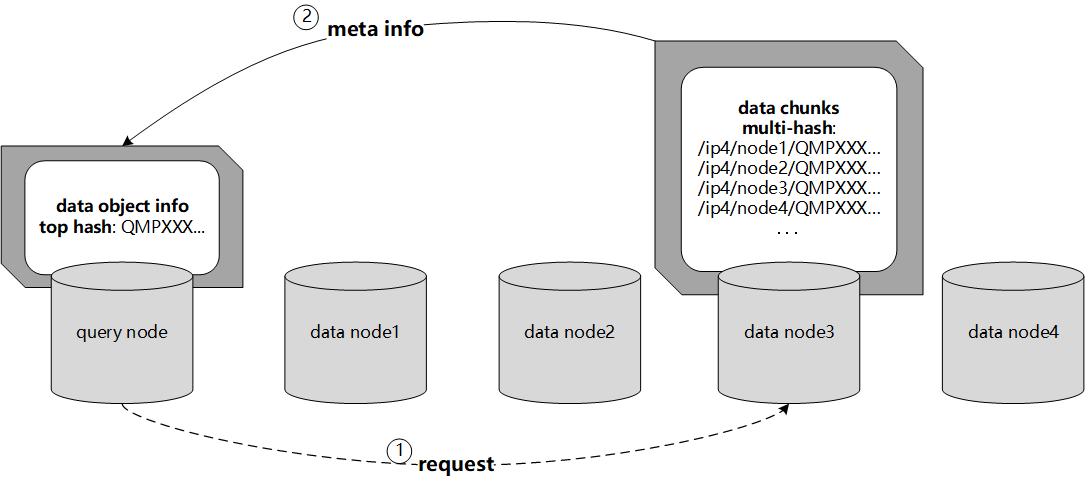}}
\caption{hash flatten of a IPFS data request}
\label{fig:DHT1}
\end{figure}

\begin{figure}[htbp]
\centerline{\includegraphics[scale=0.38]{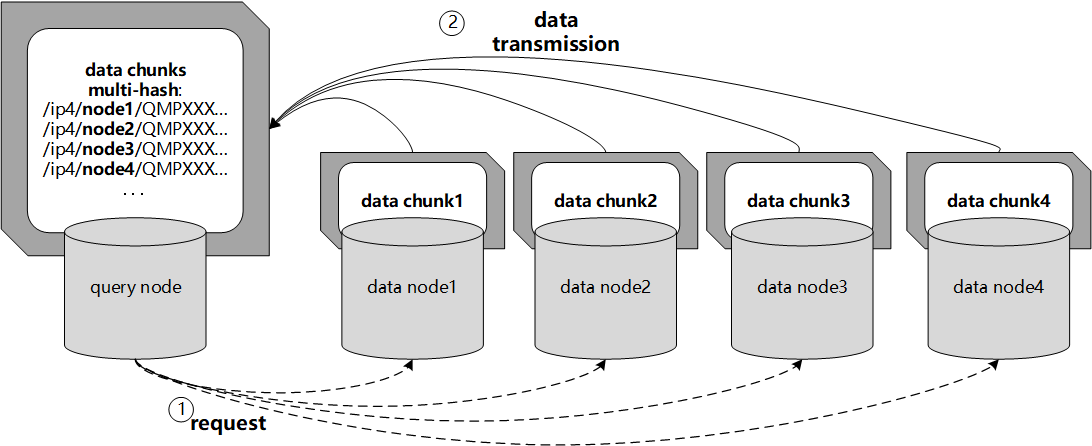}}
\caption{data transmission of a IPFS data request}
\label{fig:DHT2}
\end{figure}

IPFS uses Kademlia DHT\cite{10.5555/2541753, 22DHT:journals/corr/abs-2104-09202} to implement network routing and location requests. Its DHT stores three types of records in the form of key value pairs: 1) the mapping of data objects to their providers; 2) the mapping of nodes to their network addresses; 3) the mapping of nodes to their IPNS paths. Although DHT provides a distributed data storage scheme, the method of obtaining data location through node communication also brings high latency. One of the key efforts in Minerva's design is reducing DHT access to improve the efficiency of distributed network query.

The file chunk size specified by IPFS is $4$MB, which means that an IPFS data object has to be divided into multiple data chunks. In IPFS, all chunks of a data object are organized into a Merkle tree, where the top hash and its corresponding sub-nodes (hash) are stored in the IPFS peer that DHT can find. The leaf nodes of the Merkle tree are the hashes of each data chunk. Suppose an IPFS peer (called foreman node) need to access a specific data object from IPFS, the foreman node will find the specific IPFS peer which stores the top hash information by DHT calculating the distance between the top hash and peer CID. As Fig.\ref{fig:DHT1} shows, the children hashes information the foreman will get. When all the leaf hashes are found in cascade through the top hash, the foreman queries the DHT through these hashes to find the specific location of the chunks. Finally, the foreman could request the nodes where store the chunks to receive data chunks and form the whole data object Fig.\ref{fig:DHT2}. This indicates that for a large data object, IPFS needs at least twice DHT access to obtain the data object.

\subsection{Motivation}
The idea of Minerva stems from the combination of database science and the emerging decentralized web. Here we show two use cases that showing our motivation.

\textbf{IPFS data analytic and real-time queries.} IPFS networks have spawned a number of projects and applications. Many popular applications like Origin\cite{OriginProject} and Steepshot\cite{SteepshotProject} use IPFS as backend to store their data. These and future IPFS-based applications require real-time queries or simple analysis of data stored in IFPS to meet business needs. This is ideal for a well-developed data query system.

\textbf{Distributed federated queries over IPFS.} Because IPFS is a decentralized network, some decentralized applications that store data in IPFS or store data source information in IPFS have emerged like Mediachain\cite{MediachainProject} and Open Bazaar\cite{OpenBazaarProject}. The subordinate nodes of these applications can be considered as a data federation, and the users of the application need to query data in this data federation. Minerva provides data query tools for data federation in decentralized networks and accelerates this federated data access.

\section{System Overview}
\label{sec:overview}

Minerva is a high-performance distributed data management system over IPFS peer-to-peer networks. The architecture of Minerva is illustrated in Fig.\ref{fig:overview}, which serves as a middleware in the middle of Apache Drill query engine and IPFS storage nodes. 
Minerva consists of three key components: \emph{MinervaCoordinator}, \emph{MinervaExecutor} and \emph{MinervaCache}. 
They operate as a unity to enable federated queries on anonymous named data and improve the query performance as well. 

\begin{figure}[htbp]
\centerline{\includegraphics[scale=0.5]{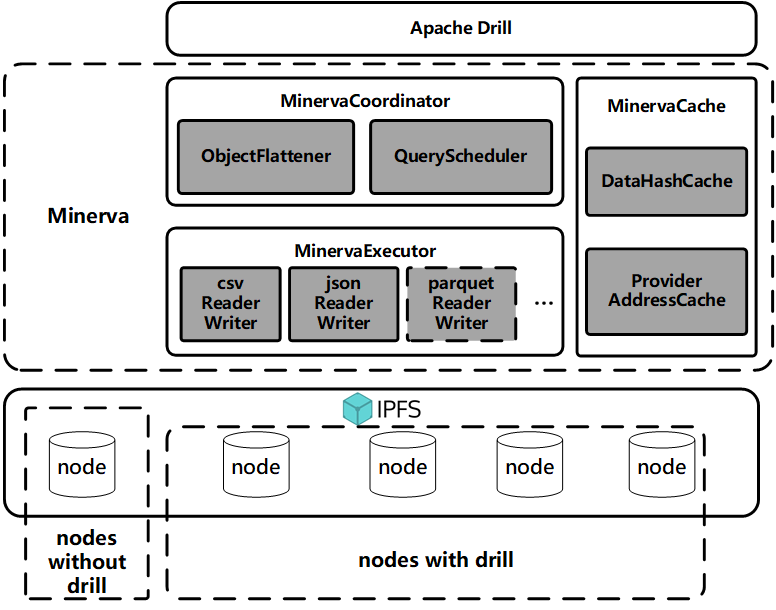}}
\caption{Architecture of Minerva Query System}
\label{fig:overview}
\end{figure}

\noindent\textbf{1) MinervaCoordinator.} The distributed query manager, \emph{MinervaCoordinator}, plays the role of assigning query tasks to anonymous nodes that store the data and generating specific data query plans. As the core of Minerva, it runs at the foreman node that launches a SQL query. MinervaCoordinator receives the Hash information of data objects parsed by Apache Drill and sends the collaborate query plans to MinervaExecutor.

\begin{itemize}
    \item \textbf{ObjectFlatterner.} The minimum unit of Minvera query is the data chunk corresponding to the leaf node of the aforementioned Merkle tree. Due to the anonymity of IPFS data objects, we designed ObjectFlatterner as a parser to extend the top hash of the data involved in its query along the Merkle tree to the hash of data chunk.

    \item \textbf{Query Scheduler.} Data objects are stored at the requested node in the form of replicas, so a data block often has multiple providers in IPFS. Query Scheduler generates and dispatches query tasks on these nodes with data block backups to coordinate multi-node collaborative queries.
\end{itemize}


\noindent\textbf{2) MinervaExecutor.} MinervaExecutor performs the execution of query plan generated by MinervaCoordinator. Each node participating in the collaboratory query runs an instance of MinervaExecutor that dynamically constructs data schema for IPFS content objects. Both the read and write functions are supported. Owing to the diverse content format such as CSV, JSON and so on, we implement specific read/write executors for each of them.

\begin{itemize}
    \item \textbf{Reader.} Reader retrieves data and parses the data into tables. This includes taking data chunks from IPFS and formatting them into relational tables in the order of record tuples. 

    \item \textbf{Writer.} Writer store the data in the format of tables into IPFS. Write operations are more complex than read operations. It partitions the data tuples in the relational table according to the \emph{maximum chunk size}, organizes these chunks into Merkle tree and writes them to IPFS.
\end{itemize}

\noindent\textbf{3) MinervaCache.} The decentralized nature of Minerva undoubtedly causes a longer query processing time compared with the centralized processing either at a local machine or a computing cluster. When a content is queried multiple times, pinning its IP addresses is time-consuming, and the query result can possibly be reused. MinervaCache is a built-in component that temporarily stores the content information to be queried, the physical locations and the object hash. 

\begin{itemize}
    \item \textbf{MetadataCache.} Merkle tree parsing of data objects is a time-consuming process, involving multiple I/O and network latency. MinervaCache caches the mapping of the top hash and its corresponding leaf nodes as the meta information of the file object to reduce frequent Merkle tree flattening.

    \item \textbf{DHT-Cache.} Peers in IPFS generally uses their ID hash as the identifier. But the actual communication needs to access their corresponding IP address, which is also a process to access DHT. Minerva uses the relevant network information of the visited peers as the DHT cache content to avoid multiple IP queries.
\end{itemize}

\section{System Design}
\label{sec:design}

In this section, we present the design of our collaborative query processing system. The functionalities of major components and their interactions are elaborated. 

\subsection{Coordinator}

Minerva makes use of IPFS, an Internet-wide decentralized file storage system, as the foundation of content hosting and routing. The core functions of Minerva's coordinator are to retrieve the IP addresses that store the content and to schedule collaborative SQL queries. Given the hash of a content, the coordinator queries it through IPFS distributed hash tables (DHT). The primary challenge to this goal is the compulsory data sharding in IPFS, i.e. the Merkle tree-based content structure. Data content is usually large while the maximum file size allowed is 4MB, so that any file exceeding this limit will be partitioned into multiple chunks. Acquiring the complete mapping between content Hashes and their physical addresses is time consuming. The second challenge is whether to split the query task across collaborative Minerva workers, and if yes how to split it under various network settings.

\begin{figure}[htbp]
\centerline{\includegraphics[scale=0.4]{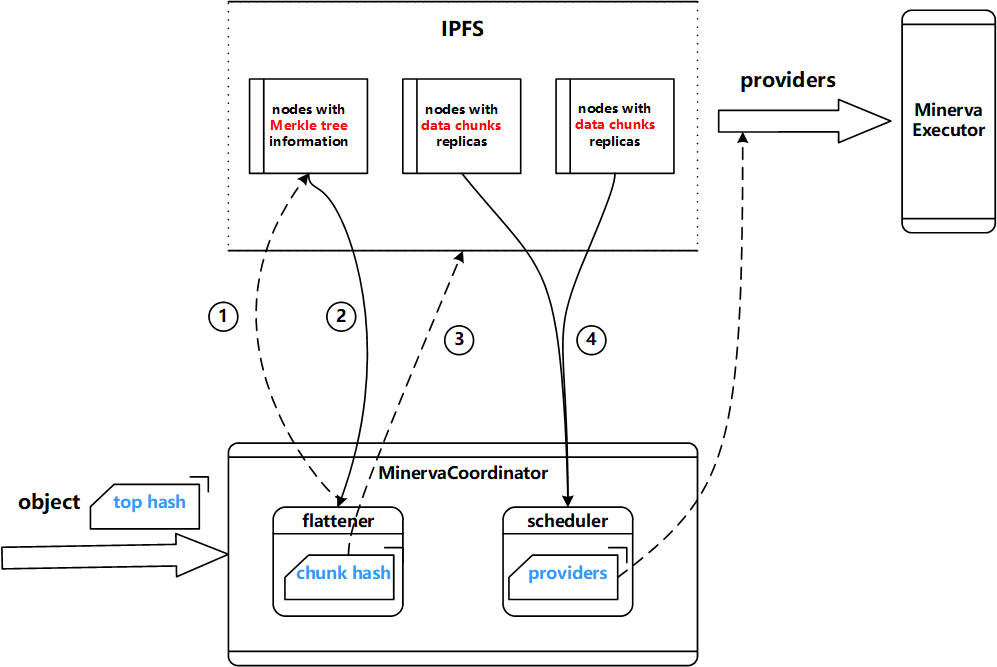}}
\caption{The workflow of MinervaCoordinator}
\label{fig:coordinator}
\end{figure}

We design MinervaCoordinator to tackle these challenges with its workflow shown in Fig.\ref{fig:coordinator}. Recall that each data shard corresponds to a leaf node of a content Merkle tree. A \emph{multi-round} DHT routing approach is developed to query all the data shards of an object. \textcircled{\small{1}} The foreman node's flattener first sends the top Hash request of the data object that corresponds to the root node of the Merkle tree to the IPFS network. \textcircled{\small{2}} After the top Hash has been discovered, the query peer returns the Hashes of the child nodes to the foreman node's flattener. If the child nodes are not the leaf nodes, the flattener repeats the DHT query. \textcircled{\small{3}} The flattener sends the Hashes of all data shards to the IPFS network. \textcircled{\small{4}} The DHT returns the list of providers for each data shard to the Minerva coordinator. The number of DHT rounds equals to the depth of the Merkle tree plus one.

To improve the efficiency of data object parsing process, we use more parallel computing methods to implement Coordinator. For example, we construct a shared thread pool for flattener execution. Each Merkle tree node will be parsed by a separate thread, allowing Merkle tree parsing to be executed in parallel. We do the same for providers of data chunk requests. We also design the cache module that supports parallel access, accelerating parsing efficiency through parallel updates and reads.

After hash parsing, Minerva's scheduler will select the appropriate providers and assign compute tasks. According to our design principles, we assign compute tasks to other peers and return the compute results instead of the data itself. Therefore, the scheduler will give priority to assign computing tasks to the nodes with Minerva. For IPFS peer without Minerva instance, we will download data as a last resort if such an authorization is granted.

Next, we will introduce and discuss the algorithm and implementation details of the above system architecture design of Minerva.

\subsubsection{Flattener}
\label{sec:details}

\begin{figure}[htbp]
\centerline{\includegraphics[scale=0.5]{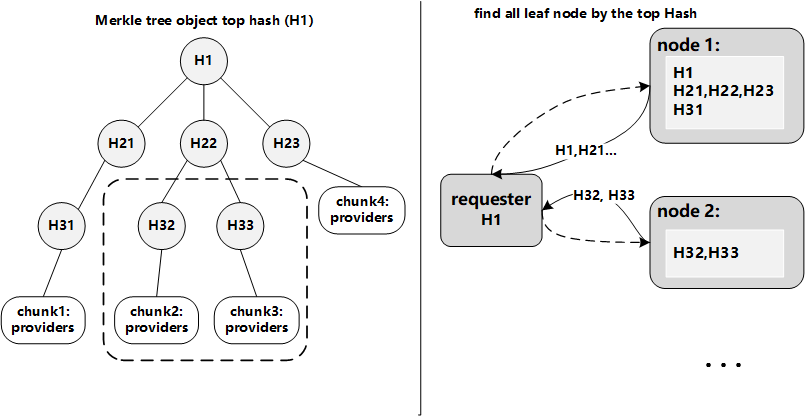}}
\caption{An object Hash parsing of IPFS}
\label{fig:flatten}
\end{figure}


We build ObjectFlattener to handle and parallelize all DHT queries. We call this process of locating queries ``Flattenning''. As Fig.\ref{fig:flatten} shows the Flattening of a data object, Minerva needs to obtain the all leaf hashes when the top hash (H1) is given. However, as shown in the example in Fig.\ref{fig:flatten}, when we have all the child hashes (H21, H22, H23) of H1, we also need to launch additional DHT queries to get the hashes of their child hashes. Even if we know the parent hash of leaf hash like H21, we still need to verify whether the child hash (H31) is a leaf. We define the function \emph{flatten} which receive a parent hash, and return its child hash and corresponding providers mapping array (by DHT query):
$$
(children, \langle providers \rangle)\  = flatten(parent) .
$$

For big data files, the amount of chunks may be large due to the limit of IPFS chunk size, which causes serious delay for Minerva to flatten the tree structure and query providers information. ObjectFlattener internally implements a self-managed shared thread pool, enabling all the DHT query work to execute parallel to relieve the delay. The flattening algorithm is shown in \ref{code1}. 

\begin{algorithm}[htbp]
    \caption{Object CID flatten}
    \label{code1}
    \hspace*{0.02in}{\bf Input:}
    topNodeHash\\
    \hspace*{0.02in}{\bf Output:}
    $(leafHash, \langle providers \rangle)$
    \begin{algorithmic}[1]
        \State $pool \xleftarrow{} new threadPool$
        \State $cache \xleftarrow{} MinervaCache$
        \State function flatten (node)
            \State $mapping \xleftarrow{} \varnothing$
            \For{$child \in node$}
                \If{$children \in child = \varnothing$}
                    \State $<providers>\xleftarrow{} IPFS.findProviders$
                    \State $mapping \xleftarrow{} mapping \cup (node, \langle providers \rangle)$
                \Else
                    \State $cached \xleftarrow{} cache.getCachedMapping(node)$
                    \If{$cached = \varnothing$}
                        \State $cached \xleftarrow{} pool.execute(flatten(child))$
                    \EndIf
                    \State $mapping \xleftarrow{} mapping \cup cached$
                \EndIf
            \EndFor
            \State \Return $mapping$
        \State $cache.put(mapping)$
    \end{algorithmic}
\end{algorithm}
We use MinervaCache to provide mapping information of data chunk and corresponding CID of providers and shared thread pool for parallel object hash flattening, which maximizes the use of processors computing resources and delays caused by network requests. Although we can parallelize the operations at the query foreman, the processing threads of queries are still may in the same machine (as H1 and H21 etc.). This may limit the parallelism width to the CPU cores of the machine, so the DHT query execution may still be partly sequential. We formulate the DHT query model in Section \ref{sec:delay anatic}.

\subsubsection{Scheduler}
Once we find all the data blocks and their locations of corresponding providers through Flatterner, the next thing we need to do is to schedule these data sources to work together. Compared with the distributed file system, the difference here is that there may be multiple data sources for each data chunk. Minerva needs a strategy to determine the query plan for each data chunk. The problem is that how to determine the workers responsible for executing each data chunk in the scheduling strategy.

When designing the scheduling strategy, we take different scenarios into consideration. For example, suppose all data blocks are held by 5 workers, it may be desirable to evenly distribute the data blocks to the worker for queries. However, the reality is that the computing power and communication bandwidth of the 5 worker nodes may be heterogeneous. Even worse, for federal queries, the computing and communication capabilities of machines may be private, and it is difficult to get this information before the query process. Even if this information is available, query nodes in the runtime may be affected by other simultaneous task loads. This makes our design extremely difficult. We design three strategies for different scenarios that work well in our experiments.

For a complete Minerva query, the data objects involved in the IPFS may be distributed in different nodes. Let $N = [n_1, n_2, ... ,n_m]$ denote $m$ IPFS nodes involved in the specific query $Q$. We can obtain the data chunks of all data objects involved in $Q$ and the corresponding possible providers by Object flattener. Let $D=[d_1, d_1, ..., d_c]$ denote all $c$ data chunks flattened by all the data objects and $P=\{p_i^j\}$ denote the $j^{th}$ available provider of $i^{th}$ data chunk. Moreover, we denote $p_i^*$ as the chosen provider of $i^{th}$ data chunk.

Since it is possible for each data chunk to have multiple providers responding, inappropriate providers selection can cause severe delays when frequent queries occur. For example, when all data chunks of an object can be provided by 3 nodes ($n_1, n_2, n_3$) (the common situation is that all three nodes have complete copies), if $n_1$ is selected to provide all data chunks, it may bring high load to the local computing resources of $n_1$ and the network links from $n_1$ to Foreman. A more reasonable providers allocation method is to evenly allocate data chunks to three providers. In this case, computing resources and network occupation can be alleviated. 

In Minerva, we propose three strategies for scheduling data chunks and their corresponding providers and embedded them as a option:

\textbf{Random Selection.}
For general data files, when the data file is not large, the nodes that pull the data can cache copies of the data. In this case, for each chunk in the file, random selection for provider can statistically achieve the global desirable performance.

$$
p_i^*= p_i^{rand(0, p_i.length)}.
$$

\textbf{Load Balance.}
In the case of large files, single node may not be able to store the whole data. At this point, balancing the computing load of each participant node is the most appropriate strategy. We assume that each node has the same computing resources. We adopt the greedy method, that each data block selects the provider with the least work load according to global information $W$ which kept by the scheduler. $W_i$ denotes the count of data chunks assigned to $n_i$.

$$
p_i^*= \underset{p_i^j}{argmin} \ W[p_i^j].
$$

\textbf{Response Priority.}
This strategy makes each data chunk preferentially select the node that responds the earliest among all providers. Such nodes may have the best data transmission environment and idle computing resources. If data is cached locally, this policy will give priority to the local selection, avoiding network. This strategy will play a key role in queries with large data transmission. 

$$
p_i^* = SortByResponse(p_i)[0].
$$

For the provider selection of each data chunk, there is no large association and dependency order between chunks. We have adopted the same shared thread pool as the Flattener, so that all provider selections can be executed in parallel to maximize the execution efficiency of this operation. We also provide a system parameter $k$ which limits the max provider count to reduce waiting time for providers response. The full provider selection process is shown in Alg.\ref{code2}.

\begin{algorithm}[htbp]
    \caption{providers selection}
    \label{code2}
    \hspace*{0.02in}{\bf Input:}
    input $D, P, k, m$\\
    \hspace*{0.02in}{\bf Output:}
    output $\langle (d_i, p) \rangle$
    \begin{algorithmic}[1]
        \State $pool \xleftarrow{} new threadPool$
        \State $coordinator \xleftarrow{} MinervaCoordinator$
        \State $selected \xleftarrow{} \varnothing$
        \State $workload \xleftarrow{} array(m)$
        \\
        \For{$i\ form\ 0\ to\  d.length$}
            \State $pool.execute(selection(i), workload)$
        \EndFor
        \State \Return $selected$
        \\
        \State function selection (i)
            \If{$P_i \notin coordinator$}
                \State $responsed = IPFS.connect(P_i, k)$
                \State $coordinator.add(responsed)$
                \State $coordinator.update(P_i)$
            \EndIf
            \State $selected \cup SelectByStrategy(d, P, i, workload)$
    \end{algorithmic}
\end{algorithm}

\subsection{Executor}

There are three roles (with data store) that the nodes involved in a file query or write request on IPFS play: 1) the query Foreman, who initiates the query and uses Apache drill as the access interface to send Minerva requests to all IPFS nodes. 2) the IPFS node with Minerva support, which can act as the query worker to execute query requests. 3) the IPFS node without Minerva support. MinervaExecutor needs to design appropriate read and write strategies for these three participants, to minimize data transmission and utilize the computing resources of all the nodes.

For the foreman node 1), MinervaExecutor directly fetches the data from the local IPFS storage and parses data file into records. In the case of 2), the node can both read and transfer data records to Foreman, or execute queries locally, such as computing the pushed down query operations. For reducing network traffic, We default to execute the calculation task and return the calculation results to Foreman in case of 2), and provide an option to directly transfer the raw data. If the node which stores unique data does not run a Minerva instance, in the case of 3), Minerva can only transfer raw data by IPFS network, and users can only use this implementation to treat IPFS as a back-end storage system without collaborative computing. In Executer, we choose foreman to recieve all the data from these kind of nodes.

We divide all queries into read queries and write queries. Read queries are SQL query requests that contain only \textbf{SELECT} statements (and various clauses), while write queries contain some SQL statements in Data Definition Language (DDL). Typical write queries, such as \textbf{CTAS (Create Table As Select)} statement and \textbf{CTTAS (Create Temporary Table As Select)} statement, are statements that need to store data in the form of relational tables in the database.

\begin{figure}[htbp]
\centerline{\includegraphics[scale=0.4]{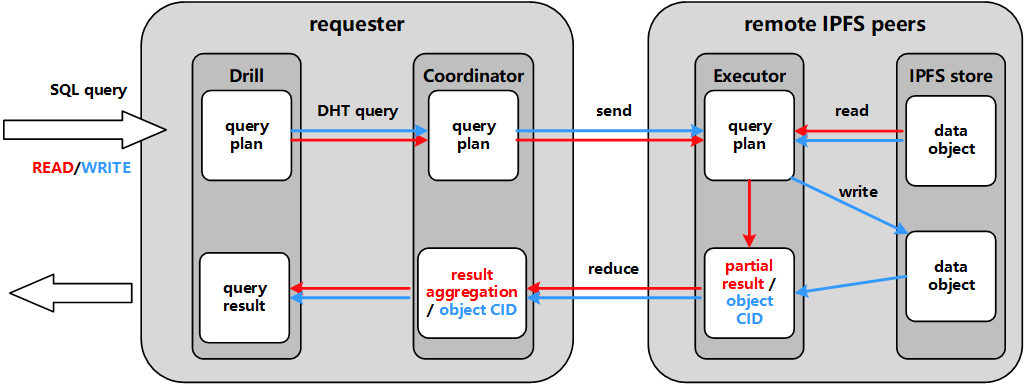}}
\caption{process of a READ/WRITE query}
\label{fig:executor}
\end{figure}

As Fig.\ref{fig:executor} show, both read and write query plans will be executed by the executor. Each executor is responsible for reading and writing local IPFS data. The read query need Executor to read the data in local store and reformat it into the form of relational table. The write query requires one step more than the read query to write a new data table into IPFS. In IPFS, this writing is in the form of new data replica, so the write operation often takes more time. Since the maximum data block size specified by IPFS is 4KB, Executor will partition the new table according to the specified block size, and organize all the data blocks into the form of Merkle tree to ensure the rationality of writing. One direct way is to slice the data file at bit-level so that the block size fits the size requirements of the IPFS data block. However, this method of file slicing requires all the data blocks to be read out and merged before they can be formatted into data records, which is inappropriate for our parallel processing and distributed computing. Therefore, writer of Minerva divides the data file into blocks at the record tuple granularity, ensuring that each block can be computed separately.

There may be many data blocks stored in one node. Our system may run out of memory if a large number of data blocks are read directly at the same time. In order to read and write data safely and reduce memory cost, MinervaExecutor first reads the boundary of a data record, and then processes each data record line by line. During the implementation of Executor, we do not adjust the implementation of IPFS, but interact with IPFS instances according to the IPFS API. This makes our executor not affect normal IPFS operations.

\subsection{Cache Design}
\label{sec:Cache Design}
As mentioned before, to reduce the query processing time, Minerva caches the CIDs and the network addresses of data providers. The goal of cache design is to preserve as much query-related information as possible while reducing the space and query time overhead of the cache itself. Minerva adds additional DHT queries to the process compared to traditional database and distributed database. So the cache of Minerva is meant to saving time for DHT queries than preserving the data results of queries.
\subsubsection{Caching File Meta-information}
The content stored in IPFS is usually organized as a Merkle tree. To acquire the information of leaf nodes, Minerva needs multiple rounds of DHT queries recursively. In our caching scheme, the meta information is cached as in a key-value store in which the key is the CID of the root node, and the value is the list of all leaf nodes related to the root CID. For a particular file, the organization and the leaf nodes of a Merkle tree is invariant. Thus, after one round of query, its structure can be cached for future use.
\subsubsection{Caching DHT Information}
In Minerva, the requester needs to know the possible data providers for each leaf node and their network addresses. Then, Minerva can distribute the query plan to these data providers. The DHT query involves rounds of DHT lookup, i.e. the lookup of the root CID and the lookup of the sub-CIDs. Therefore, the DHT lookup delay throttles the efficiency of decentralized query processing. We can choose to cache these DHT query information that avoids the repetitive DHT lookup on the same set of CIDs. The freshness of the cached DHT is crucial to the query performance in Minerva. The cached mapping between CIDs and the network addresses are usually invalid if the data provider becomes offline or deletes the data shards. In these situations, blindly bypassing the DHT lookups may cause the requester to send the query plan to invalid providers, causing a very long timeout delay. Therefore, Minerva setups a time-to live (TTL) for all the cached DHT items. Afte the TTL timer expires, the requester can remove the cached DHT items or initiates a new round of DHT lookup.
\subsubsection{Update mechanism}
When new entry need to be added to cache, and entries in the cache exceeds the cache buffer size, the least recently used (LRU) entry will be popped out of the cache to make room for new entry. When a new entry is added, the write time is recorded in the cache. In this way, when the cache hits, the cache will check the generation time of the hit entry. If the existence time exceeds the effective time, the entry is considered to have expired, and will be evicted from the cache.

\section{DHT Delay Reduction}
\label{sec:delay anatic}

In this section, we present a probabilistic model of the DHT query delay that accounts for a major part of the end-to-end latency. A fat Merkle tree structure is proposed to effective reduce the DHT query delay.



\subsection{DHT Query Model}

We characterize Minerva's DHT query delay that starts from sending out the top hash query till obtaining the provider list of all data chunks. Consider $N$ data chunks randomly distributed on $M$ nodes of an IPFS network, where $M$ is sufficiently large. All the data chunks are organized as the leafs in a perfect $k$-ary Merkle tree so that the height of this Merkle tree, $H$, is computed as:
\begin{eqnarray}
H = \left \lceil   \log_k N \right \rceil  + 1.
\end{eqnarray}

The total number of hashes on this Merkle tree is given by:
\begin{eqnarray}
\tilde{N} = \frac{k^{H}-1}{k-1} = \frac{k^{\left \lceil   \log_k N \right \rceil  + 1} - 1}{k-1} \approx \frac{k}{k-1}N.
\end{eqnarray}
Let $T$ be the total DHT query delay consisting of two parts, where one is the time of flattening the Merkle tree from the root to the leafs (ialso ncluding the leafs), i.e. $T_{flt}$, and the query time to acquire the provider list for each leaf hash (i.e. $T_{pro}$).

Kademlia DHT protocol is adopted in IPFS systems so that the total DHT latency is the sum of the delays of querying hashes on the Merkle tree. Let $t^{KAD}$ be the delay of a DHT query that is a random variable with the probability density function $f(t)$ and the cumulative density function $F(t)$, $(t\geq 0)$. The probability distribution of the query time in a specific Kademlia DHT is computed in \cite{2008Faster}, and the query time is also upper bounded by $O(\lceil \log M \rceil)$ in \cite{Cai_2013}. Since our goal is not to model the exact DHT query latency, we only require that the DHT query time is independent and identically distributed (a.k.a. \textit{iid}). 

When carrying out DHT queries, we should be clear that the nodes at a Merkle tree are flattened in parallel. An interesting observation is that the flatterning times of the nodes at the same layer are \emph{independent} and identically distributed, while the flatterning times of an ancestor node and a descendant node are \emph{dependent}. The reason is rather intuitive: the flatterning time of the ancestor node includes that of his descendant nodes. Denote by $T_{flt}^{h,j}$ the flatterning time of the $j^{th}$ node and its child nodes at the $h^{th}$ layer of the Merkle tree. The following recursion relationship holds:
\begin{eqnarray}
T_{flt}^{h,j} = t^{KAD}_{h,j} + \max_{l\in(k(j-1), kj]}{T_{flt}^{h+1,l}}, \;\;\; \forall h\in[1, H]. 
\label{eqn:dht_model}
\end{eqnarray}
Note that at the $H^{th}$ layer, i.e. the leaf layer, the flatterning time $T_{flt}^{H,j}$ is exactly $t^{KAD}$. The max operator indicates that the flatterning time of the $j^{th}$ node of the $h^{th}$ layer is the sum of his DHT query time and the longestest flatterning time of his children. For simplicity, we define $T^{h,j}_{m} = \max_{v \in ((j-1)k, jk]} T^{h+1,v}_{flt}$.Then the  probability density
function of $T^{h,j}_{flt}$ could be expressed as the recursive formula:
\begin{align}
    f_{T^{h,j}_{flt}}(t) &= \int_{0}^{t} f_{T^{h,j}_{m}}(x) f(t-x) \ dx,\\
    f_{T^{h,j}_{m}}(t) &= k (F_{T^{h+1,jk}_{flt}}(t)) ^ {k-1} f_{T^{h+1,jk}_{flt}}(t).
\end{align}

Due to the query time of all leaf nodes at the $H^{th}$ layer is \emph{iid}, $T^{H-1,j}_{m} = \max_{v \in ((j-1)k, jk]} T^{H,v}_{flt} = \max_{v \in ((j-1)k, jk]} t^{KAD}_{H,v}$. The distribution of $T^{H-1,j}_{m}$ is expressed as:
\begin{eqnarray}
    f_{T^{H-1,j}_{m}}(t) = k (F(t)) ^ {k-1} f(t).
\end{eqnarray}

\subsection{Fat Merkle Tree}
Through the above analysis, we find that both the intermediate nodes and the leaf nodes will be queried by Minerva through DHT query, resulting in a lengthy hash flattening time. 
We propose a fat Merkle tree index structure that is more suitable for network data query. The total cost of DHT queries is reduced from the perspective of intermediate nodes and leaf nodes respectively. We expect the fewer intermediate nodes of Merkle tree, so that the fewer nodes need to be extended. However, we still need to retain the distributed storage feature of Merkle tree fast hash verification as much as possible. The structure tree  tends to be flat as much as possible. For leaf node DHT access, it is necessary to confirm whether its node identity is leaf (the parser does not know whether a child hash is a leaf node). If the identity of leaf nodes can be confirmed from the perspective of structure, most DHT queries can be reduced. 

\begin{figure}[htbp]
\centerline{\includegraphics[scale=0.45]{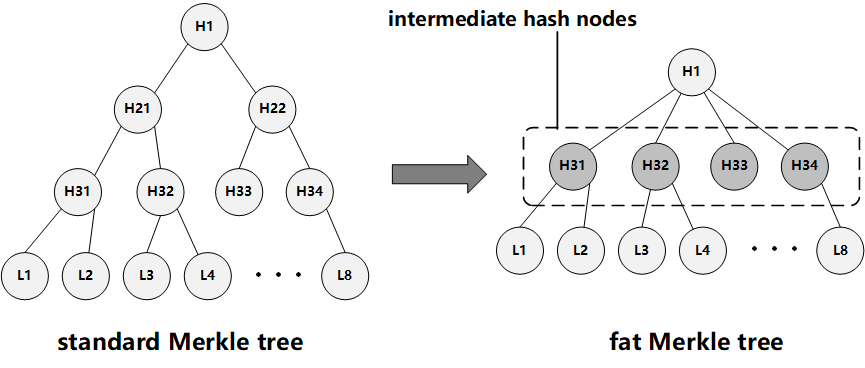}}
\caption{Structure of fat Merkle tree}
\label{fig:FMT}
\end{figure}

Our distributed file index structure protocol (fat Merkle tree structure) is very simple but effective as Fig.\ref{fig:FMT} shows: we limit an IPFS data object index Merkle tree up to three levels, including top node ( CID of the data object), intermediate node layer and leaf node layer. Each node in the middle layer points to $k$ child nodes (leaf nodes), while the top node connects all intermediate nodes. This is equivalent to flattening the middle layer of the original Merkle tree structure by compressing it all into one layer, and the parent nodes are all top nodes. At the same time, because our structure only has up to three levels of index, we can confirm that the third level node must be a leaf node, which omits the third level node query. In our structure, the total time of DHT hash flattening is:
\begin{align}
    T^{FMT}_{flt} &= T_{flt}^{1} + T_{flt}^{2} \nonumber\\
    &= t^{KAD}_{1,1} + \max_{j\in[1, \frac{N}{k}]} t^{KAD}_{2,j}.
    \label{eqn:dht_FMT}
\end{align}

The probability density function of total query time $T^{FMT}_{flt}$ can be expressed as:
\begin{eqnarray}
    f_{T^{FMT}_{flt}}(t) &= \int_{0}^{t} (\frac{N}{k} (F(x))^{\frac{N}{k}-1} f(x))(f(t-x)) dx.
\end{eqnarray}

\begin{theorem}
For a certain file with $N$ chunks, the flattening time of standard Merkle tree is longer than fat Merkle tree.
\end{theorem}

\begin{proof}
Let $T^{STD}_{flt}$ be the flattening time of standard Merkle tree and $T^{h,j}_{max}$ be as $T^{h,j}_{m} = \max_{l \in ((j-1)k, jk]} T^{h+1,l}_{max}$. $T^{H-1,j}_{max}=\max_{l \in ((j-1)k, jk]} t^{KAD}_{H,l}$. From the equation \ref{eqn:dht_model},
\begin{align}
T_{flt}^{H-1,j} = \max_{l\in(k(j-1), kj]}{T_{flt}^{H,l}} = T^{H-1,j}_{max}.
\label{eqn:dht_proof_1}
\end{align}

From \ref{eqn:dht_proof_1}, we could easily the following result by deduction:
\begin{eqnarray}
    T_{flt}^{h,j} &\geq T_{max}^{h,j} + t_{h,j}^{KAD} \;\;\; \forall h \in [1, H-1] \nonumber\\
    E_{T_{flt}^{h,j}} &\geq E_{T_{max}^{h,j}} + E_{t^{KAD}} \;\;\; \forall h \in [1, H-1] 
    \label{eqn:dht_proof_2}
\end{eqnarray}

Because $T_{max}^{1,1}$ means the maximum time of all the leaf nodes, we can rewrite $T_{max}^{1,1}$ as $T_{max}^{1,1} = \max_{j\in[1,N]} t^{KAD}_{H,j}$, the mathematical expectations of $T_{max}^{1,1}$ is:
\begin{align}
    E{}E_{T_{max}^{1,1}}&=\int_{0}^{+\infty} xN(F(t))^{N-1}f(t)\;dt\nonumber\\
    &>\int_{0}^{+\infty} x\frac{N}{k}(F(t))^{\frac{N}{k}-1}f(t)\;dt
    \label{eqn:dht_proof_3}
\end{align}

Combine \ref{eqn:dht_proof_2}, \ref{eqn:dht_proof_3} and \ref{eqn:dht_FMT}, we could get the result as:
\begin{align}
    E_{T^{STD}_{flt}}=E{T_{flt}^{1,1}} &\geq E_{T_{max}^{1,1}} + E_{t^{KAD}} \nonumber \\
     &\geq \int_{0}^{+\infty} x\frac{N}{k}(F(t))^{\frac{N}{k}-1}f(t)\;dt + E_{t^{KAD}} \nonumber \\
     &= E_{T^{FMT}_{flt}}
\end{align}

\end{proof}

It should be noted that our data structure does not sacrifice the hash verification ability of Merkle tree. We use the middle layer of the tree as the hash check layer. In our case, there is no need for frequent data verification, so the fat merkle tree structure is more suitable for our data query scenario.

For general Merkle tree data structures, converting indexes into the new structure described above by path compression is not complicated. We know that the nodes on each Merkle tree are mapped by the hashes of their child nodes. Thus, we only need to hash-map the the nodes in the layer $\lceil \log_k N \rceil$ of the original Merkle tree together and use the resulting hash as the top hash of the new structure. The cost of this transformation is $\mathcal{O}(\frac{N}{k})$ and is trivial in data organization

\subsection{Providers Solving Reduction}
In order to get the exact location of each data chunk, the request for the provider list is not negligible, but it also causes a significant delay. To solve this problem as much as possible, we analyze it from the point of view of data distribution. We find that for data files with fewer data blocks, the lower probability of distributed storage they with.

In solution of Minerva, we pay special attention to data with small size. Such data queries can often be completed by a single machine. We assume that data objects less than the threshold $\alpha$ are stored on the node where the top hash is located. This assumption allows us to omit the request of the provider list and directly send the query execution plan to the node, we call this greedy provider resolving. If the node with top hash does not store all data chunks, Minerva will continue to request the provider list for all chunks and reschedule their execution. So the delay of provider seeking becomes:
$$
    T_{pro} = \begin{cases}
    \max_{m=1}^{M} t^{KAD} + \delta, & \textrm{discentral\ store,}\\
    0, & \textrm{central\ store.}
    \end{cases}
$$
where $\delta$ refers to the time cost for the remote node to discover and return errors when assuming errors. However, the node only needs to return the control information when it finds that the data is not stored locally. So $\delta$ is small, and the time cost caused by the error is trivial.

\section{Implementation}
\label{sec:implementation}
In this section, we introduce the implementation details of Minerva. Instead of building an individual process communicating with Apache Drill and IPFS, we implement Minerva as a storage plugin of Drill to exploit the query planning mechanism of Drill. 
In order not to make any changes to IPFS, we interact with the IPFS daemon to operate the data in IPFS, and these interactions are completed with IPFS-JAVA-API. The Minerva system is written in JAVA, consisting of 3150 lines of core module and 1800 lines of object format interface module. 

\begin{figure}[htbp]
\centerline{\includegraphics[scale=0.35]{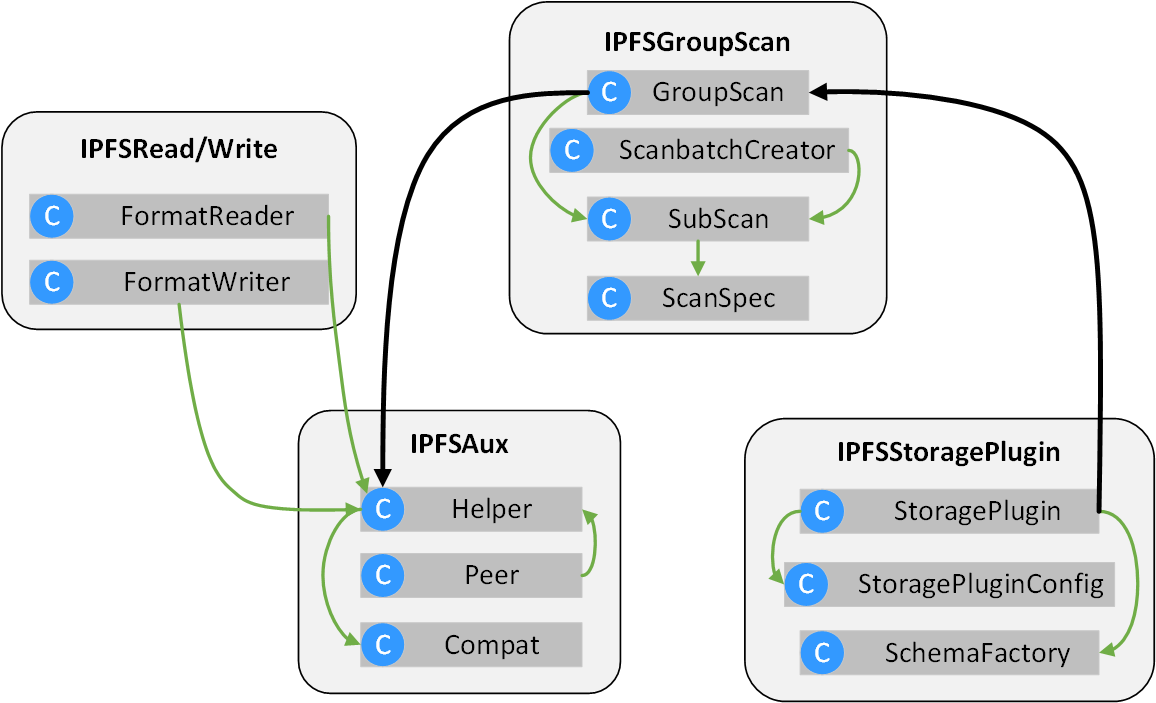}}
\caption{modules of Minerva}
\label{fig:UML}
\end{figure}

The code of Minerva is form into four main modules as Fig.\ref{fig:UML} shows: (1) \textit{IPFSStoragePlugin} defines the storage interface for Drill; (2) \textit{IPFSGroupScan} can as a kernel module of Minerva generates and expresses the query plan; (3) \textit{IPFSRead/IPFSWrite} is in charge of reading and writing local data with various form. For all the data formats, we build their specific loading rules and algorithm into a sub-class (4) \textit{MinervaAux} contains a set of functions to configure IPFS clients for communication, generate the schema of a query at runtime, and provide operational contexts interconnecting different modules.

As the entrance of other operations, \textit{IPFSStoragePlugin} incorporates the basic objects supporting the operations of Drill and IPFS. It implements the \textit{StoragePlugin} function that specifies the interactions between the IPFS storage module and the query engine modules.More exactly, \textit{IPFSStoragePlugin} inherits the class \textit{AbstractStoragePlugin} and overrides some methods at parent classes including \textit{supportsRead}, \textit{supportsWrite} and \textit{registerSchemas}, etc.

The workflow of \textit{IPFSStoragePlugin} operates as follows. The daemon process of Drill (Drillbit) first initializes a \textit{StoragePluginRegistry} object that records all storage plugins and configures the mappings between the storage backends and their corresponding plugin realizations. \textit{StoragePluginRegistry}
then scans Minerva’s \textit{IPFSStoragePlugin} and loaded it in a modular manner. \textit{StoragePluginRegistry} creates an object of
\textit{IPFSStoragePlugin} class (via createPlugins), and calls the constructor (via newInstance) to complete the initialization of Minerva’s storage plugin. At last, \textit{StoragePluginRegistry} launches IPFSStoragePlugin that paves the way of Internet wide query processing.

\textit{IPFSGroupScan} consists of \textit{IPFSGroupScan} and
\textit{IPFSSubScan} classes. Here, \textit{IPFSGroupScan} is the underlying logic operator defined in Minerva. It scans the entire dataset and reads data from the storage backend to the memory so that other data processing. And it makes the global decision of DHT parsing and strategy. \textit{IPFSGroupScan} generates \textit{SubScan} operators to enable the parallel query processing and to capture the data dependency of subtasks. Each time a requested initiates a query processing task, an \textit{IPFSGroupScan} object is created, and the root CID of the dataset stored in IPFS is obtained. After recursively expanding the root CID, \textit{IPFSGroupScan} obtains the network addresses of the providers across Internet. All the above information is crucial to generating an \textit{IPFSSubScan} class.

\textit{IPFSSubScan} is the sub-operator of \textit{IPFSGroupScan}, representing the query operations executed by a provider on its data stored in IPFS. \textit{IPFSSubScan} includes all the information needed by a successful query execution, such as the CIDs of data shards and data formats stipulated by SQL. The \textit{IPFSGroupScan} class can be serialized and transferred to all the working nodes as a part of the query plan, and deserialized to acquire the information regarding to query execution.

\textit{IPFSAux} module consists of \textit{IPFSHelper}, \textit{IPFSCompact} and \textit{IPFSPeer}. The main functionality of this module is to interconnect different components that offer the compatibility with underlying IPFS operations and the encapsulation. \textit{IPFSHelper} provides the encapsulation of underlying IPFS APIs. \textit{IPFSCompact} is designed to handle the IPFS requests with IPFS instance. \textit{IPFSPeer} is an abstraction of an IPFS node in the network, containing the node ID and the network address.


For the specification of data queries, we define Minerva's input and output format as a generic data table. And thanks to the schema-free feature of drill, we do not need to specify a schema for the data in advance which ensures anonymity. Users only need to specify the basic data file format to query a IPFS object. The query format of the sample Minerva is as code \ref{code:SQLQueryExample}:
\begin{figure}[h]
\begin{minted}[fontsize=\scriptsize,breaklines]{sql}

select * from ipfs.
`/ipfs/QmdfTbBqBPQ7VNxZEYEj14VmRuZBkqFbiwReogJgS1zR1n#csv`;

select * from ipfs.`/ipfs/top hash#format`;

select * from ipfs.`/ipfs/top hash/example.format`;
\end{minted}
\caption{SQL query example of Minerva}
\label{code:SQLQueryExample}
\end{figure}


\section{Evaluation}
\label{sec:evaluation}
In this section, we presents the evaluation of Minerva. We first introduce our experiments setup, followed by the evaluations of delay reduction and MinervaCache. We further present the overall performance of Minerva on different network environments and query types.


\subsection{Setup}
\label{ex_setup}
\subsubsection{Methodology}
We conduct our experiments on 8 Intel(R) Xeon(R) 3.20GHz machines, with each machine equipping 16 CPU cores and 32 GB memory. The bandwidth of each machine is 10 Gbps in cluster mode and degraded to 1 Gbps and 100 Mbps when simulating the WAN environments. Minerva as a back-end service system, has strict latency requirements for service in real application, so we use query latency as a performance evaluation metric in the experiment. A complete query includes two stages: query plan generation and query plan execution, which are completed by MinervaCoordinator and MinervaExecutor respectively. We also analyze the two stage in the experiment.

We deploy the same Minerva instance on each experimental machine, and all machines join the same IPFS private network to build a distributed Minerva cluster. Because Apache drill (SQL parser for Minerva) can accept requests in the form of REST interfaces, we randomly select the foreman node used to launch queries, and use REST form requests to continuously input query tasks. After all queries are completed, we analyze the statistical query performance through the profile.

\subsubsection{Dataset}
\begin{itemize}
    \item OAG dataset\cite{inproceedingsOAG}: Open Academic Graph (OAG) is a public dataset of knowledge mapping containing 80GB data, which records 150 million academic papers and their reference relationships. We used OAG dataset  to assess Minerva query performance in real-world environments.
    \item TPC-DS dataset\cite{inproceedingsTPC}: TPC-DS is an evaluation framework for database performance test, which covers simple and complex SQL queries in a variety of real scenarios. TPC-DS will generate test datasets and corresponding query statements. This framework is used in our experiments to evaluate Minerva's performance on different types of queries. We set the scale factor to 10 and turn the generated dataset to JSON format.
    \item Synthetic dataset: We build synthetic datasets to simulate the specific data structure, different count of data chunks and records. The synthetic datasets range in size from 1MB to 1GB. We constructed the synthetic data structure in JSON format that meets our test requirements and are uploaded it to IPFS in accordance with our index structure format.
\end{itemize}

\subsection{DHT Delay Reduction}
\label{ex_DHT_relay}
We conduct an experimental analysis on DHT parsing of data, the most critical step in network query for Minerva system. We set the chunk size to 1MB, and as the number of chunks changes, the parsing time for the top hash should also change.
\begin{figure}[htbp]
  \centering
  \subfigure[DHT latency reduction methods]{
  \label{fig:ex-DHT1-a} 
  \begin{minipage}{.22\textwidth}
  \centering
    \includegraphics[scale=0.28]{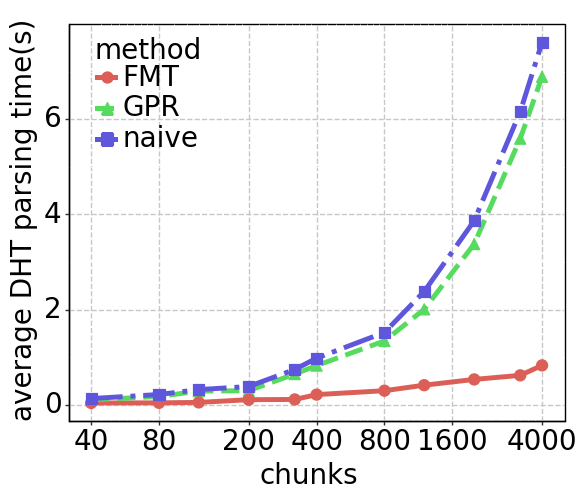}
  \end{minipage}
  }
  \subfigure[Overall latency reduction effect]{
  \label{fig:ex-DHT1-b} 
  \begin{minipage}{.22\textwidth}
  \centering
    \includegraphics[scale=0.28]{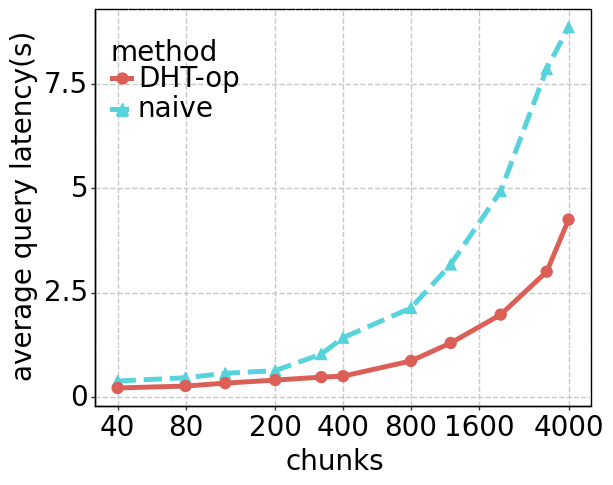}
  \end{minipage}
  }
   \caption{DHT latency reduction}
  \label{fig:ex-DHT1} 
\end{figure}

As shown in Fig. \ref{fig:ex-DHT1}, we test the two methods of reducing DHT parsing time described in Section.\ref{sec:delay anatic}, namely fat Merkle tree (FMT) and greedy provider resolving (GPR). The overall DHT query latency of Minerva increases with the increase of chunks. This is because the more chunks, the more Merkle tree nodes need to be parsed, and each node parsing may require network queries.

After Minerva uses the DHT reduction method like FMT and GPR, the overall DHT query time has decreased compared to without it (native) 
 Fig.\ref{fig:ex-DHT1-a}. When there are fewer chunks, neither of the two optimization methods significantly optimizes DHT latency, because fewer DHT query operations for the fewer chunks and the process of DHT query is highly parallelized. As the number of chunks becomes large, the optimization effects of the two methods are gradually evident. GPR is approximately $5-10\%$ less than native DHT queries. FMT decreases by about $80\%$, and the reduction effect increases with the number of chunks, which is consistent with our analysis. The GPR method reduces the provider resolution time of the leave hash, which is proportional to the amount of chunks. While FMT reduces the parsing of leaf nodes and almost all the intermediate nodes in the Merkle tree, resulting in a significant reduction in the final DHT time.

In Fig.\ref{fig:ex-DHT1-b}, we also show the overall query latency of Minerva with both DHT reduction method. In our evaluation, we only recorded the query latency when the data was first accessed to avoid being interfered by multilevel caching. The results show that the optimized Minerva improves significantly over the native DHT query, reducing the overall latency by about 50\%. However, this improvement decreases with the number of nodes in the data block, because as the number of chunks increase, query execution time gradually dominates the cost of the entire query, and DHT improvements become less apparent.

\begin{figure}[htbp]
  \centering
  \subfigure[Impact of tree height on DHT]{
  \label{fig:ex-DHT2_a} 
  \begin{minipage}{.2\textwidth}
  \centering
    \includegraphics[scale=0.26]{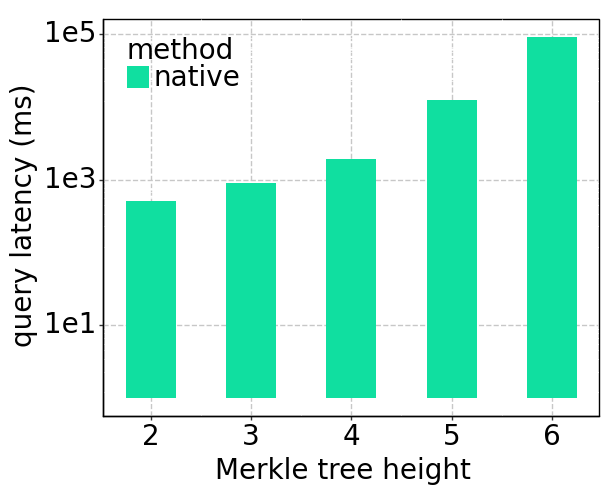}
  \end{minipage}
  }
  \subfigure[Rescheduling cost of DHT-GPR]{
  \label{fig:ex-DHT2_b} 
  \begin{minipage}{.25\textwidth}
  \centering
    \includegraphics[scale=0.26]{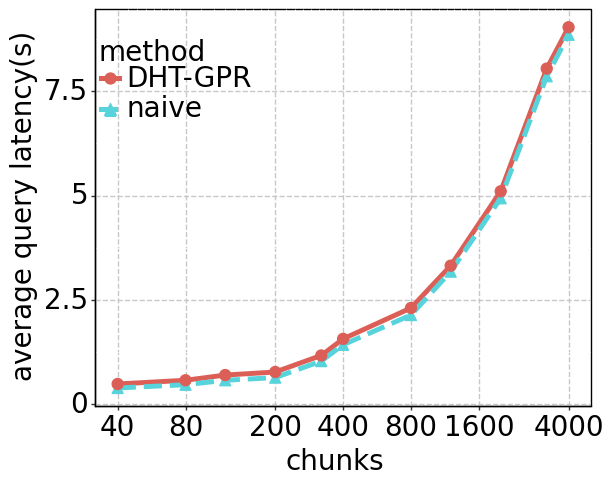}
  \end{minipage}
  }
   \caption{DHT latency reduction}
  \label{fig:ex-DHT2} 
\end{figure}

To confirm the effect of DHT delay reduction in our experiment, we test the parsing performance of Minerva on various Merkle tree structures. As shown in Fig.\ref{fig:ex-DHT2_a}, we change the tree height of the final Merkle tree by changing the chunk size. Minerva is sensitive to structural changes in the Merkle tree. As the tree height increases, query latency increases exponentially. This is because as chunks increases, the intermediate nodes on the Merkle tree also increases. This result can also explain why FMT is more effective than GPR.

In order to explore the additional overhead of GPR, we show the GPR and the query delay of native Minerva in Fig.\ref{fig:ex-DHT2_b} when the assumptions are incorrect. We find that the additional cost of Minerva correcting errors and re-performing DHT parsing of the top hash is almost a fixed value when the assumption is incorrect (that is, when all data is not located on the top hash node). In our experiment, this value is about $0.1-0.2s$, which is not significant compared to the overall query latency and GPR gain for Minerva.

\subsection{Cache Performance}

\begin{figure}[htbp]
\centerline{\includegraphics[scale=0.26]{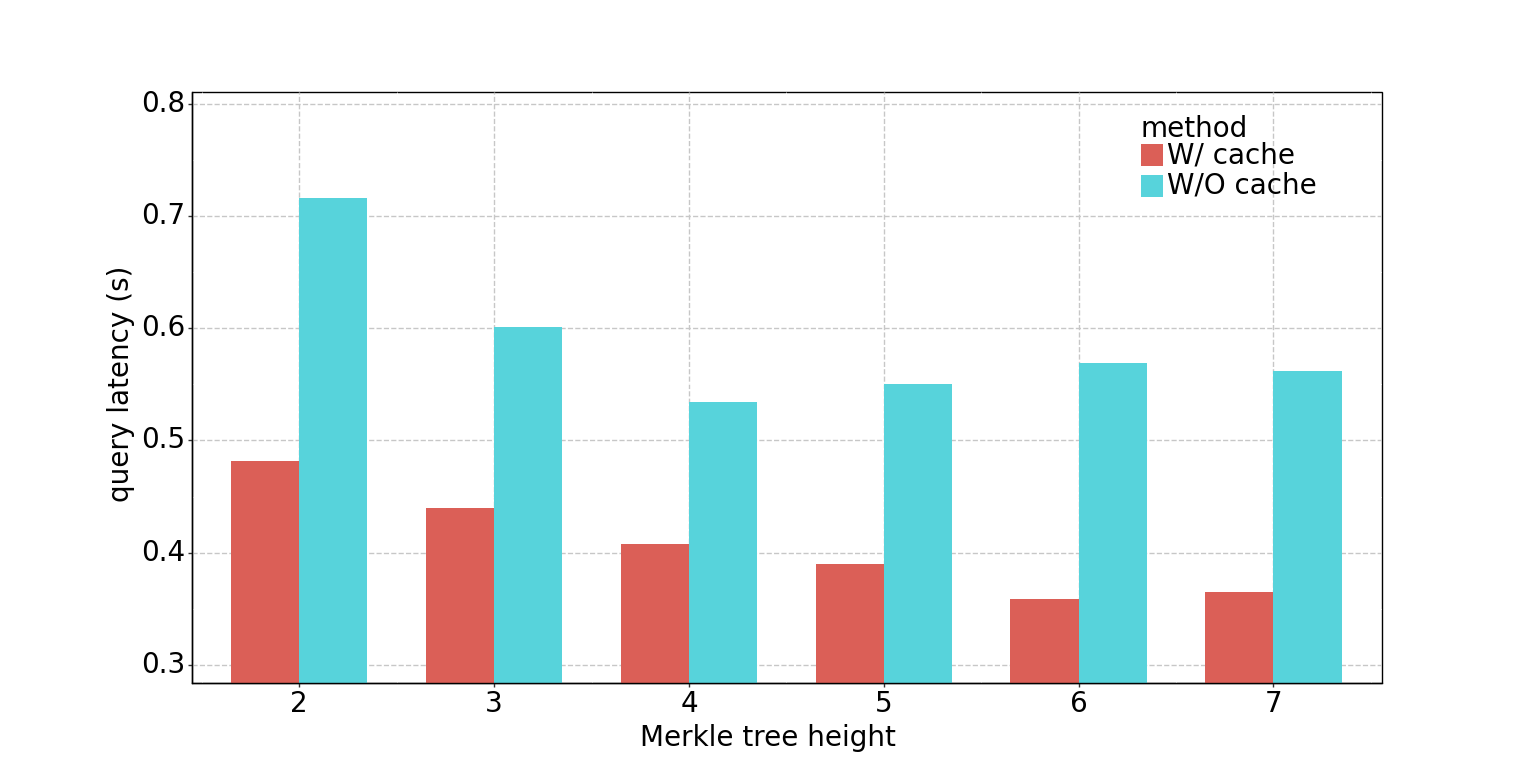}}
\caption{Overall effect of MinervaCache}
\label{fig:cache}
\end{figure}

\label{ex_cache}
We hereby evaluate MinervaCache, the second method used by Minerva to reduce data query delays on IPFS. As shown in Fig.\ref{fig:cache}, we test the query performance improvement using MinervaCache. Compared to the case where cache is not enabled, the enabling of MinervaCache reduces the average query latency by $50-60\%$, and the effect increases as the parallel execution nodes increase. At the same time, we observe that without using cache, query latency does not decrease when four nodes execute in parallel. This is because as the nodes increases, the pressure on DHT top hash parsing and planning also increases, so the overall delay does not decrease. MinervaCache greatly reduces this impact, allowing Minerva to efficiently query with a larger parallel width.

\begin{figure}[htbp]
\centerline{\includegraphics[scale=0.26]{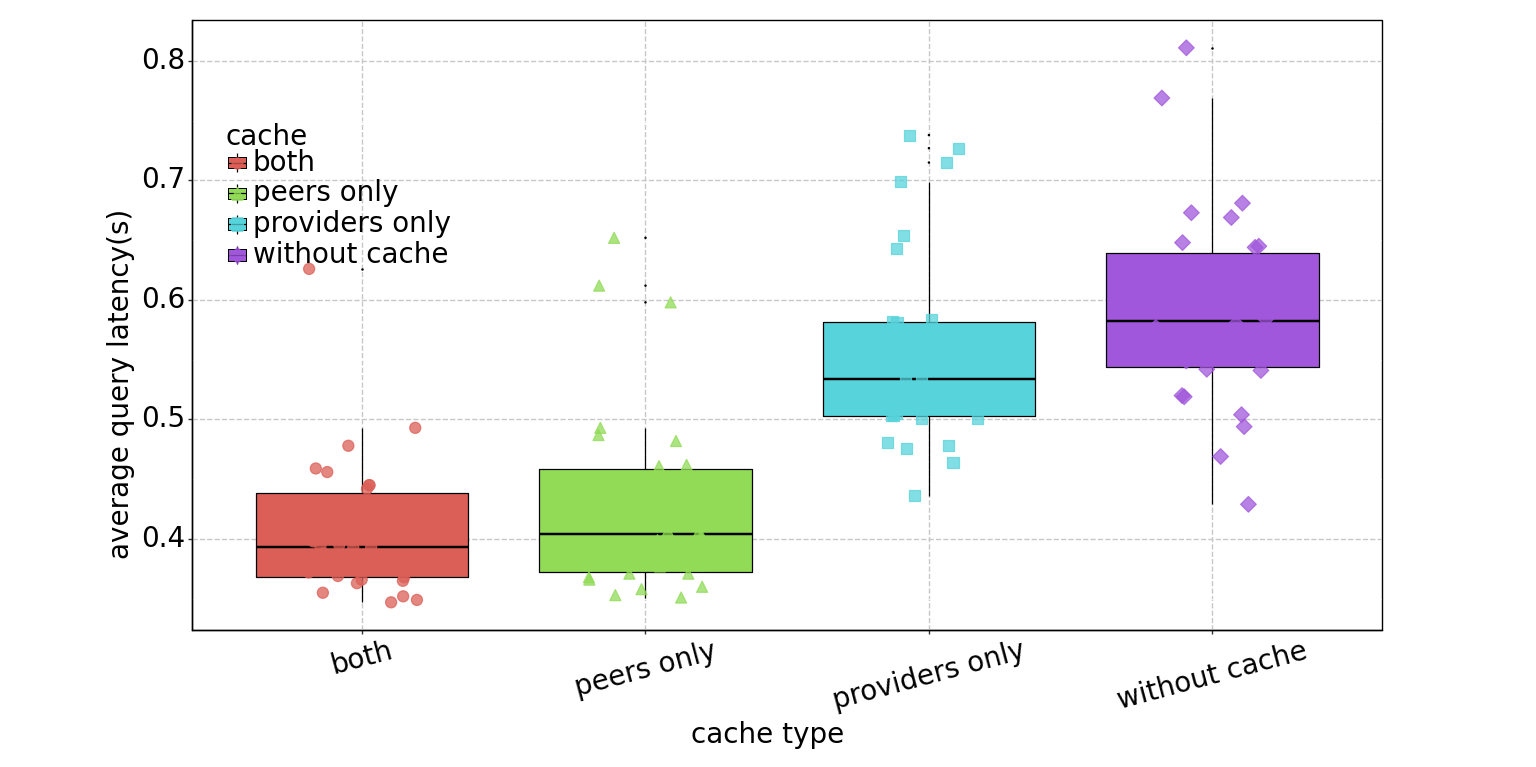}}
\caption{component of MinervaCache}
\label{fig:cache analysis}
\end{figure}

We further compare two parts of MinervaCache, namely, the cache for data chunk providers and the physical address cache for known IPFS peers. As shown in Fig.\ref{fig:cache analysis}, the caching strategy of providers only has a performance improvement of about $9\%$ compared to no caching, while the caching strategy of peers only has a latency reduction of nearly $30\%$. At the same time, the query results of providers only strategy has a larger variance than those of peers only strategy. This is because the cache of providers is for data chunks, and the data objects queried each time are not always the same, so the cache hit rate is low. The peers' information is relatively stable and confined in a small set so that the cache hit ratio is very high.


\begin{figure}[htbp]
\centerline{\includegraphics[scale=0.26]{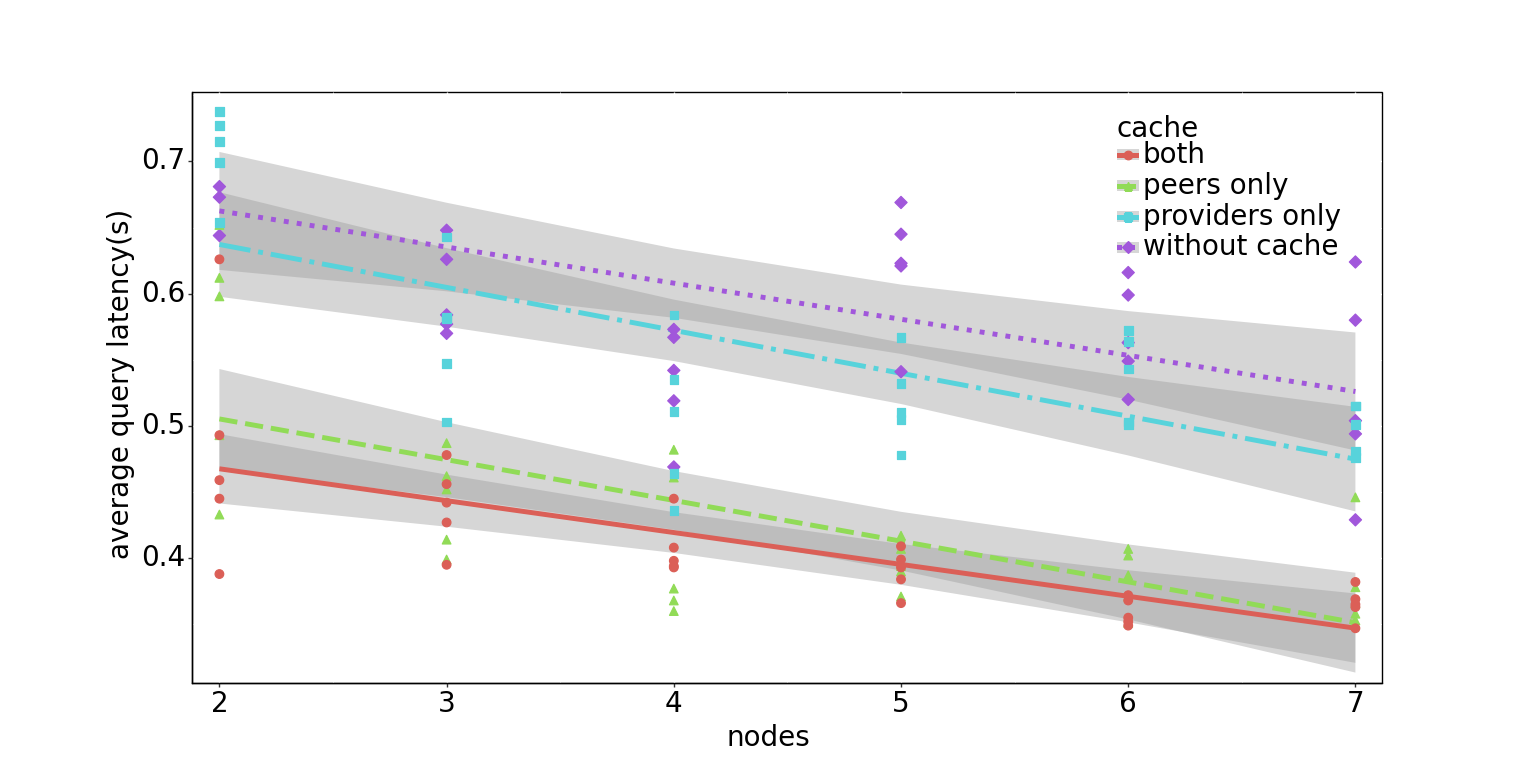}}
\caption{Impact of cache effect on average query latency}
\label{fig:cache analysis 2}
\end{figure}

We also simply fit the effect of MinervaCache with its components. As shown in Fig.\ref{fig:cache analysis 2}, with the increase in the number of execution nodes, the caching effects of both providers only and peers only strategies gradually become apparent. This is because as the parallel nodes increase, the more peer caches need to be used, and the more complex the data distribution of data chunks becomes. At this point, the provider's cache strategy will gradually play a greater role.

\subsection{System Parameter Searching}
\label{ex_params}
A complete data file stored in IPFS needs to be divided into smaller data chunks and organized into Merkle tree structures. For a query, we can assume that the data is divided into data chunks of size $s$. Since the size and partitioned data chunks both have influence on  query planning and query execution, the factor $s$ significantly affect system performance. We performed a parameter search on $s$, and the results are shown in Fig.\ref{fig:param_a}.

\begin{figure}[htbp]
\centerline{\includegraphics[scale=0.35]{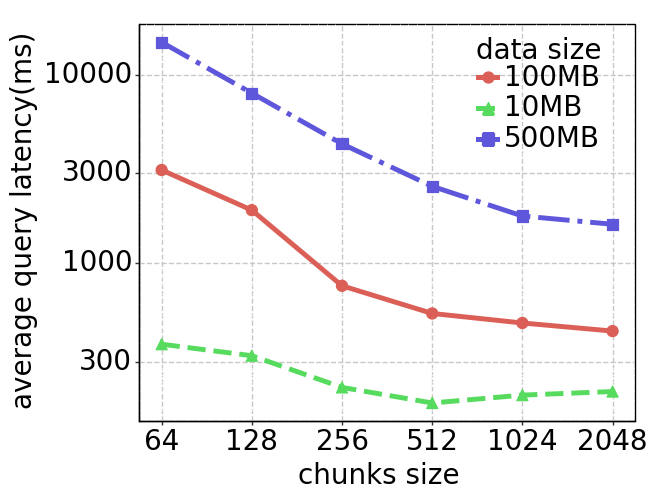}}
\caption{Impact of chunk size on average query latency}
\label{fig:param_a}
\end{figure}

\begin{figure}[htbp]
\centerline{\includegraphics[scale=0.26]{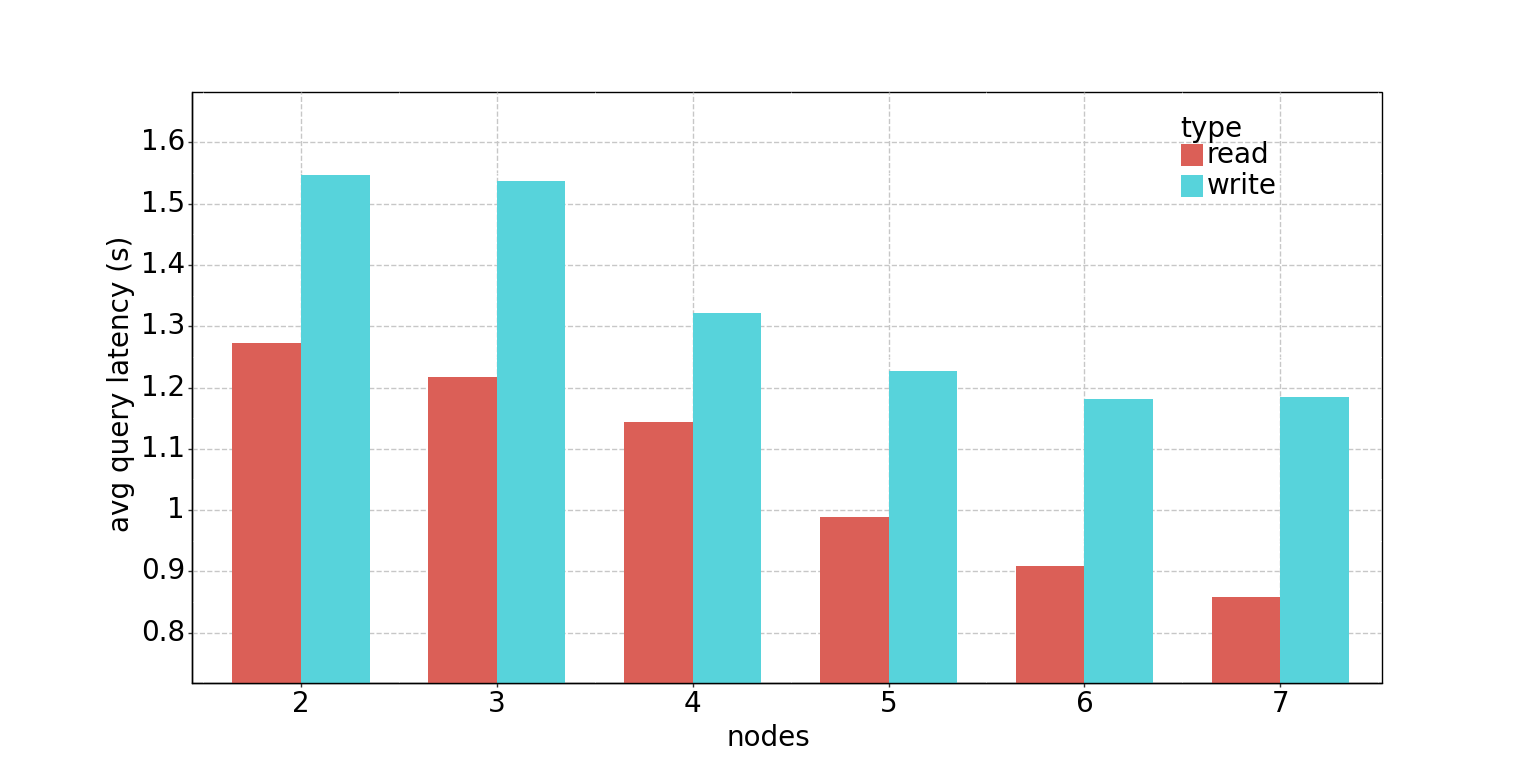}}
\caption{Impact of chunk size on read/write latency}
\label{fig:param_b}
\end{figure}

\begin{figure*}
  \centering
  \subfigure[OAG dataset]{
  \begin{minipage}{.22\textwidth}
  \centering
    \includegraphics[scale=0.28]{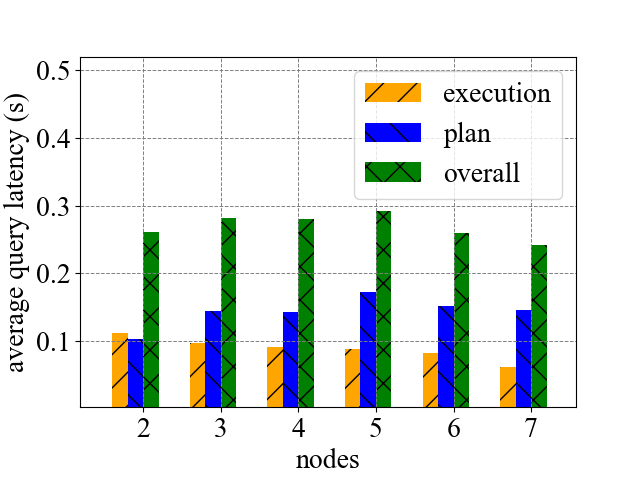}
  \end{minipage}
  }
  \subfigure[TPC-DS dataset]{
  \begin{minipage}{.22\textwidth}
  \centering
    \includegraphics[scale=0.28]{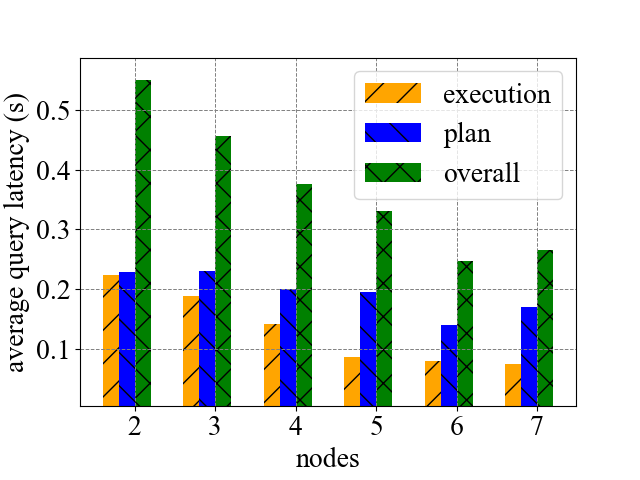}
  \end{minipage}
  }
   \subfigure[OAG 1Gbps]{
  \begin{minipage}{.22\textwidth}
  \centering
    \includegraphics[scale=0.28]{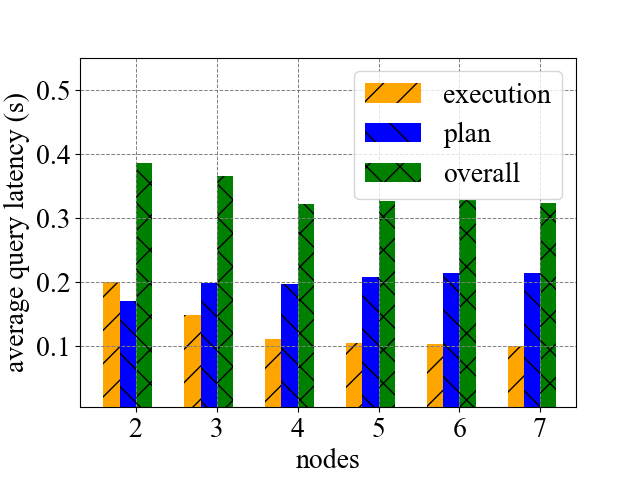}
  \end{minipage}
  }
  \subfigure[TPC-DS 1Gbps]{
  \begin{minipage}{.22\textwidth}
  \centering
    \includegraphics[scale=0.28]{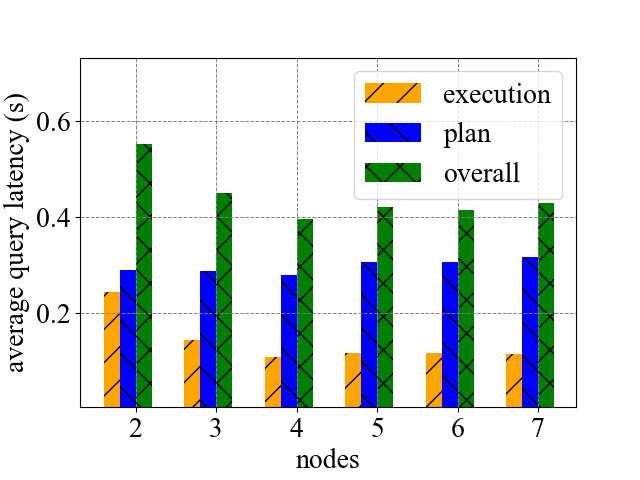}
  \end{minipage}
  }
  \caption{Overall query latency}
  \label{fig:overall performance} 
\end{figure*}

The overall latency of Minerva decreases first and then increases with the increase of $s$.  We builds data tables with sizes of 10MB, 100MB, and 500MB for testing. The overall query latency of Minerva gradually decreases with the increase of $s$. This is because as $s$ increases, DHT parsing operations gradually decrease, which greatly reduces the time required for top Hash parsing and provider resolution. However, when the chunk size is large, if the data table is small (10MB), the final query time will not always decrease, or even increase. Because query execution time dominates at this time, and the increase in $s$ also increases execution overhead and reduces the degree of parallelism, resulting in a decrease in overall query efficiency at this time.

We also test the relationship between the execution time of read/write queries and $s$ as shown in Fig.\ref{fig:param_b}. Read queries are defined in our experiment as queries other than 'CRAETE TABLE' and 'CREATE TEMPORARY TABLE'. Write query is the query that creates a new data table using the above statements. Write query of Apache drill does not modify existing data in place, but does "write" by creating a modified copy of the original data. Therefore, write queries are often based on a read query, and contain a read query execution process. To show the difference between two different queries, we selected data tables with a large amount of data (large than 500MB) to do the evaluation. As the Fig.\ref{fig:param_b} shows, write queries are 200ms slower than read queries, regardless of the nodes, which is the time spent by the write operation alone. In addition, writing a query requires generating a new table content identifier using the newly generated data chunks. The more data blocks there are, the more time it takes.

The result shows that too large or too small chunk size will have a negative impact on the performance. For read queries, we need to adjust the parallelism width to achieve the better performance and write queries may need fewer nodes to achieve the best.

\subsection{Overall Performance}
\label{ex_overall}
We store the data randomly and evenly in the nodes participating in the queries, and randomly select the Foreman node for data query. As shown in Fig.\ref{fig:overall performance}, Minerva can complete common requests in less than 400ms, as well as complex distributed query tasks in 600ms. This is efficient in the case of distributed big data query and limited network bandwidth. As the number of nodes increases, the query time should increase due to the more dispersed data distribution and the higher transmission overhead. Minerva also shows its scalability to nodes. When the number of nodes grow, it uses distributed computing to reduce latency, so that latency does not increase significantly or even decrease. 

As Fig.\ref{fig:overall performance} shows, with the increase of the nodes participating in parallel execution, the query execution time of both datasets decreases gradually at first, and then tends to be flat; When the parallel nodes keep increasing, the execution efficiency decreases. In general, the more nodes, the higher computing performance can be utilized by using more resources; However, the time spent on the entire query is also limited by the slowest node in the cluster. Therefore, increasing the number of nodes cannot reduce the query time indefinitely. Even in the case of more nodes (such as the 7 nodes in the figure), the increased communication and scheduling overhead will offset the benefits of parallel execution.

Compared to query execution time decreasing with the number of nodes, the query planning time increases with the nodes count. This is because the current query planning stage is only carried out at the Foreman node, and adding participating node will not enable the query planning stage to be completed in parallel; on the contrary, with the increase of nodes, the feasible nodes that Foreman nodes need to consider when planning queries increases. The network addresses of these nodes need to be resolved, and the connectivity of these nodes needs to be checked.

\begin{figure}
  \centering
   \subfigure[OAG dataset]{
  \begin{minipage}{.22\textwidth}
  \centering
    \includegraphics[scale=0.28]{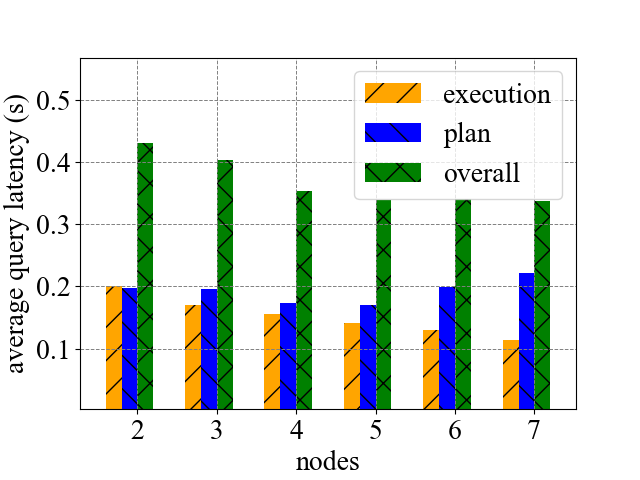}
  \end{minipage}
  }
  \subfigure[TPC-DS dataset]{
  \begin{minipage}{.22\textwidth}
  \centering
    \includegraphics[scale=0.28]{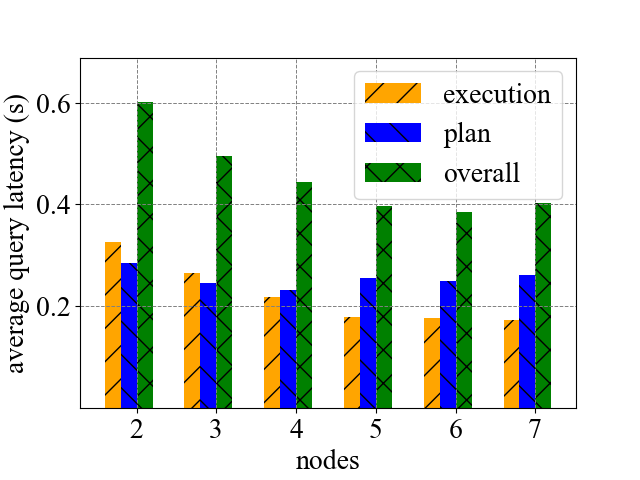}
  \end{minipage}
  }
  \label{fig:overall performance with 100 Mbps} 
  \caption{Overall query latency with 100 $Mbps$ network}
\end{figure}

Therefore, we find that since the plan time increases and the execution time decreases with the parallel nodes, the overall latency may decrease first and then increase. When perform query on a specific dataset, we should to select an appropriate parallel width to achieve the best query efficiency.

\section{Conclusion}
\label{sec:conclusion}
Data island problems hinder the development of data analysis and the application of big data. In this paper, we propose the first system to process federated data query based on IPFS. IPFS is a solution for data sharing and decentralized data storage, and Minerva's design provides a standard database style access paradigm for this data storage. In order to realize federated data query, we use Apache drill as the query engine, and design a query collaborator that adapts to the decentralized file hash table of IPFS. Moreover, Minerva also implements executors in various file forms and caches to meet query requirements. Extensively experiments show that Minerva performs well in cluster networks with 1Gbps bandwidth and wide area networks with 100Mbs bandwidth.


\bibliographystyle{IEEEtran}
\bibliography{ref}{}

\begin{IEEEbiography}[{\includegraphics[width=1in,height=1.25in,clip,keepaspectratio]{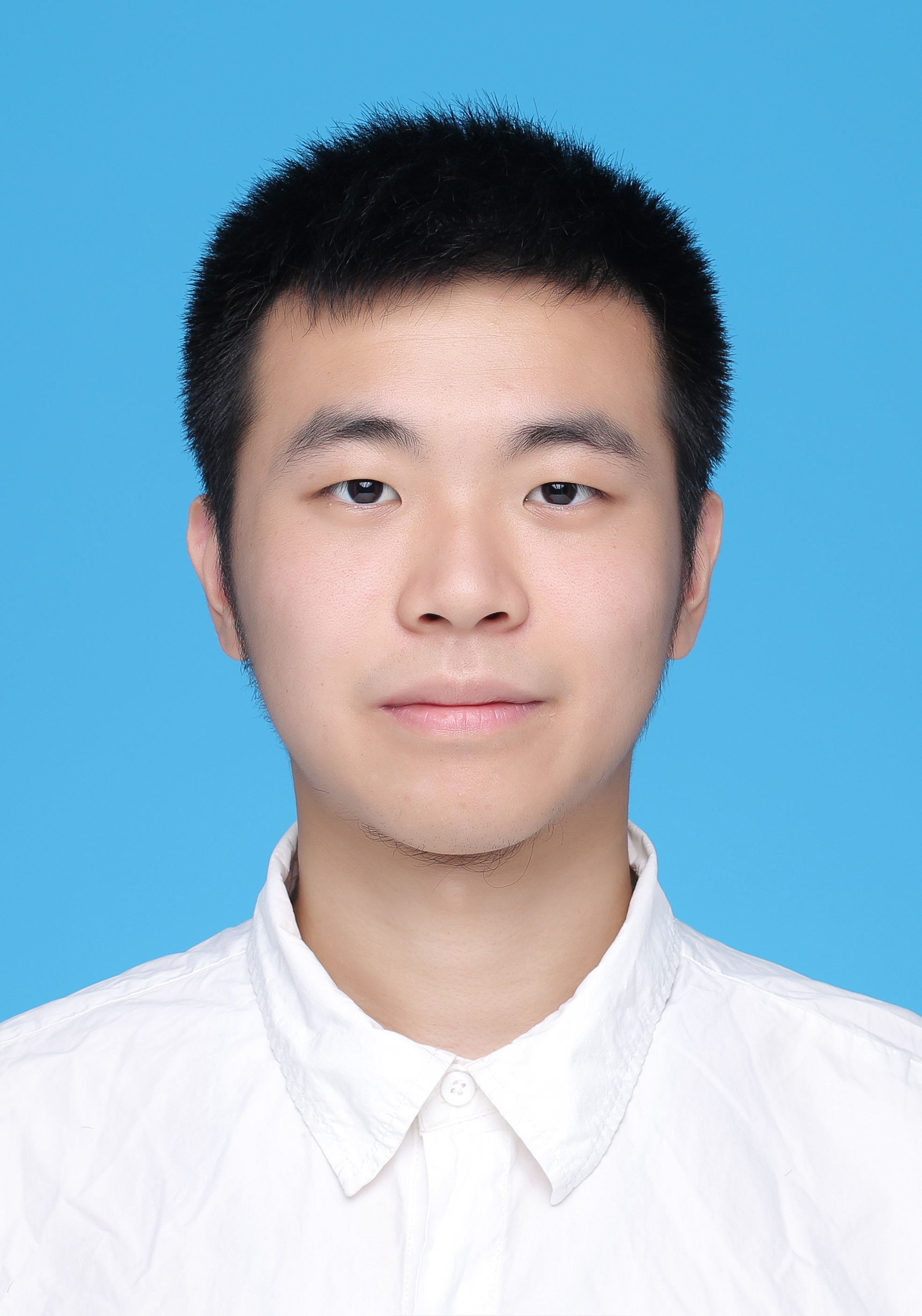}}]{Zhiyi Yao}
		is a Ph.D. candidate in electronic science and technology at School of Information Science and Technology, Fudan University, China. His research interests include distributed data analysis, deep learning training and inference system. 
	\end{IEEEbiography}

	\vspace{-1cm}

 \begin{IEEEbiography}[{\includegraphics[width=1in,height=1.25in,clip,keepaspectratio]{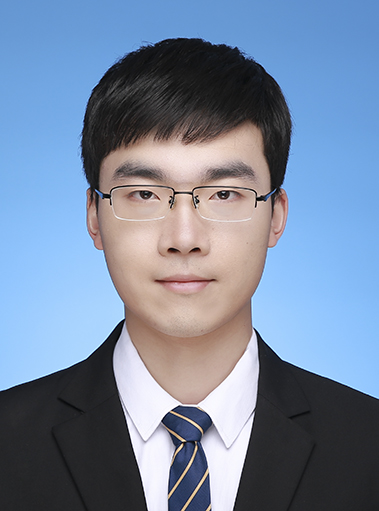}}]{Bowen Ding}
		was a graduate student at Fudan University. He is now a software engineer at Alluxio, Inc. His research insterests include distributed caching and storage systems.
	\end{IEEEbiography}

	\vspace{-1cm}

\begin{IEEEbiography}[{\includegraphics[width=1in,height=1.25in,clip,keepaspectratio]{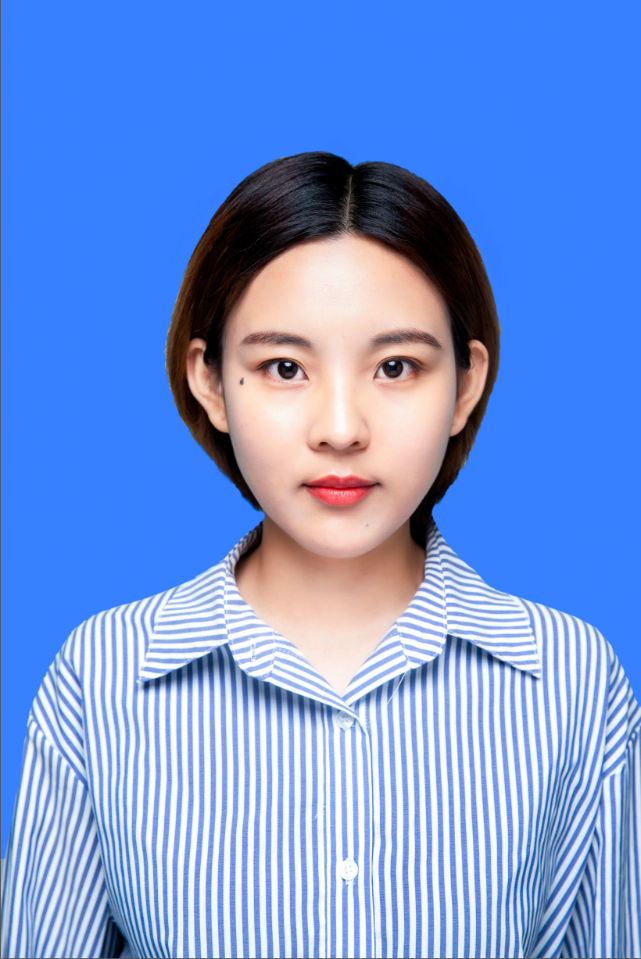}}]{Qianlan Bai}
		is a Ph.D. candidate in
		computer science at Shanghai Key Laboratory of Intelligent Information Processing, School of Computer Science, Fudan University, China. Her research interests include economic analysis and security analysis about blockchain. She has published a few papers in IEEE TIFS, TNSE, etc. 
	\end{IEEEbiography}

	\vspace{-1cm}

	\begin{IEEEbiography}[{\includegraphics[width=1in,height=1.25in,clip,keepaspectratio]{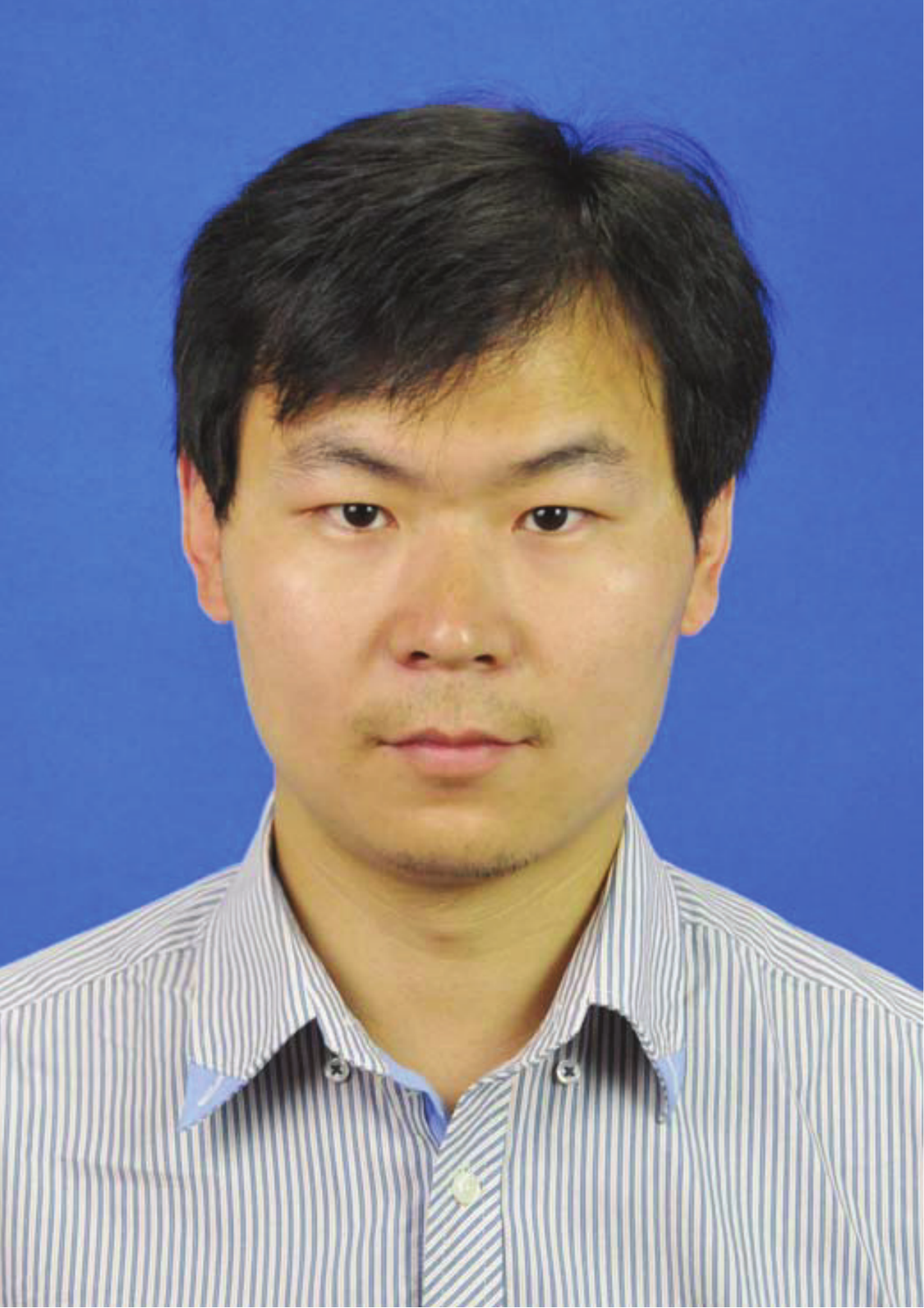}}]{Yuedong Xu}
		is a Professor with School of Information Science and Technology, Fudan University, China. He received B.S. degree from Anhui University in 2001, M.S. degree from the Huazhong University of Science and Technology in 2004, and Ph.D. degree from The Chinese University of Hong Kong in 2009. From 2009 to 2012, he was a Postdoctoral Researcher with INRIA Sophia Antipolis and Université d'Avignon, France. He received the French MENRT fellowship in 2009, the OKAWA Foundation research grant in 2019. His areas of interests include network performance evaluation, multimedia networking and distributed machine learning. He has published more than 30 papers in premier conferences and journals including ACM Mobisys, CoNEXT, Mobihoc, IEEE Infocom and IEEE/ACM ToN, IEEE JSAC. His areas of interest include performance evaluation, 
		optimization, data analytics and economic analysis of communication networks and distributed machine learning systems.

\end{IEEEbiography}

\end{document}